\documentclass[10pt,journal,twocolumn,a4paper]{IEEEtran}
\usepackage{amsfonts,amsmath,amssymb,amstext,amsthm,cite,color,dsfont,enumerate,hyperref,makeidx,verbatim,xfrac,xr,url}
\usepackage{pict2e}
\usepackage[mathscr]{euscript}
\usepackage[geometry]{ifsym}
\usepackage{tikz}
\usetikzlibrary{decorations.pathreplacing, positioning, calc}

\hypersetup{hidelinks,
pdftitle		={The Sphere Packing Bound for DSPCs with Feedback a la Augustin},
pdfauthor		={Bar{\i}s Nakiboglu},
}
\theoremstyle{plain}
\newtheorem{lemma}{Lemma} 
\newtheorem{theorem}{Theorem}

\newtheorem*{conjecture*}{Conjecture}

\theoremstyle{definition}
\newtheorem{definition}{Definition}

\newtheorem{remark}{Remark}
\newtheorem*{remark*}{Remark}
\definecolor{mygray}{gray}{0.4}
\newcommand{\set} [1]			{{\mathscr{{#1}}}}
\newcommand{\alg}[1]			{{\mathcal{{#1}}}}
\newcommand{\rndv}[1]      {{\mathsf{{#1}}}}

\newcommand{\msr}[1]       {{\it    {{#1}}}}
\newcommand{\cnst}[1]      {{\mathit{{#1}}}}

\newcommand{\sss}[1]		{{\mathit{2}^{{#1}}}}
\newcommand{\integers}[1]	{{\mathbb{Z}}_{^{{#1}}}}

\newcommand{\reals}[1]		{{\mathbb{R}}_{^{{#1}}}}
\newcommand{\bigo} [1]     {{\cnst{O}\left({{#1}}\right)}}

\newcommand{\pmf}[0]       {{p.m.f.~}}

\renewcommand{\vec}[1]     {\overrightarrow{{#1}}}

\newcommand{\dif}[1]       {{\mathrm{d}{#1}}}  
\newcommand{\der}[2]        {\tfrac{\dif{#1}}{\dif{#2}}}

\newcommand{\DEF}[0]			{{\!\!~\triangleq\!~}}  
\newcommand{\mtimes}[0]			{{\circledast}}

\newcommand{\abs}[1]           {{\left\lvert{{#1}}\right\lvert}}

\newcommand{\lon}[1]           {{{\left\lVert{{#1}}\right\lVert}}} 

\newcommand{\ind}[0]           {{\imath}}
\newcommand{\jnd}[0]           {{\jmath}}
\newcommand{\knd}[0]           {{\kappa}}
\newcommand{\tin}[0]           {{\cnst{t}}}
\newcommand{\blx}[0]           {{\cnst{n}}}

\newcommand{\PXS}[2]         {{\bf P}_{{#1}}\!\left[{#2}\right]}
\newcommand{\EXS}[2]         {{\bf E}_{{#1}}\!\left[{#2}\right]}

\newcommand{\PX}[1]          {\PXS{\!}{{#1}}}                      
\newcommand{\EX}[1]          {\EXS{\!}{{#1}}}                      
\newcommand{\PCXS}[3]        {\PXS{{#1}}{\left.\! {#2} \right| {#3}}}
\newcommand{\ECXS}[3]        {\EXS{{#1}}{\left.\! {#2} \right| {#3}}} 

\newcommand{\PCX}[2]         {\PCXS{\!}{{#1}}{{#2}}}  
\newcommand{\ECX}[2]         {\ECXS{\!}{{#1}}{{#2}}}  

\newcommand{\fX}[0]          {{\cnst{f}}}

\newcommand{\gX}[0]          {{\cnst{g}}}

\newcommand{\drv}[0]			{{\rndv{D}}}

\newcommand{\grv}[0]			{{\rndv{G}}}

\newcommand{\hrv}[0]			{{\rndv{H}}}

\newcommand{\RD}[3]				{{\cnst{D}}_{{#1}}            \!\left(\left.            \! {#2}\right\Vert {#3}                  \right)}

\newcommand{\RMI}[3]			{{\cnst{I}}_{{#1}}            \!\left(                  \! {#2};         \!{#3}                \!\right)}

\newcommand{\RC}[2]				{{\cnst{C}}_{{#1},{#2}}}

\newcommand{\RR}[2]				{{\cnst{S}}_{{#1},{#2}}}

\newcommand{\cset}[0]			{{\set{A}}}

\newcommand{\spe}[1]			{{\cnst{E}_{sp\!}}       \left({#1}\right)}

\newcommand{\rate}[0]			{{\cnst{R}}}

\newcommand{\rnb}[0]          {{\cnst{\beta}}}
\newcommand{\rnf}[0]          {{\cnst{\phi}}}

\newcommand{\rng}[0]          {{\cnst{\rho}}}
\newcommand{\rno}[0]          {{\cnst{\alpha}}}
\newcommand{\rnt}[0]          {{\cnst{\eta}}}

\newcommand{\Pem}[1]           {{\it P_{{{\bf e}}}^{{#1}}}}         
 
\newcommand{\enc}[0]           {{\varPsi}} 
\newcommand{\dec}[0]           {{\varTheta}}    

\newcommand{\brl}[0]           {{\alg{B}}}
\newcommand{\rborel}[1]        {{\brl}({#1})}

\newcommand{\oev}[0]           {{\set{E}}}

\newcommand{\pmea}[1]          {{{\alg{P}}({#1})}}

\newcommand{\pdis}[1]          {{{\set{P}}({#1})}}

\newcommand{\dinp}[0]          {{\cnst{x}}}
\newcommand{\inp}[0]           {{\rndv{X}}}
\newcommand{\inpS}[0]          {{\set{X}}}
\newcommand{\inpA}[0]          {{\alg{X}}}

\newcommand{\dout}[0]          {{\cnst{y}}}
\newcommand{\out}[0]           {{\rndv{Y}}}
\newcommand{\outS}[0]          {{\set{Y}}}
\newcommand{\outA}[0]          {{\alg{Y}}}

\newcommand{\dsta}[0]          {{\cnst{z}}}
\newcommand{\sta}[0]           {{\rndv{Z}}}
\newcommand{\staS}[0]          {{\set{Z}}}

\newcommand{\dmes}[0]          {{\cnst{m}}}
\newcommand{\mes}[0]           {{\rndv{M}}}
\newcommand{\mesS}[0]          {{\set{M}}}

\newcommand{\dest}[0]          {{\widehat{{\cnst{m}}}}}
\newcommand{\est}[0]           {{\widehat{{\rndv{M}}}}}

\newcommand{\mA}[0]				{{\msr{a}}}    
\newcommand{\amn}[1]			{{{\mA}_{{#1}}}}

\newcommand{\mB}[0]				{{\msr{b}}}

\newcommand{\mP}[0]				{{\msr{p}}}    

\newcommand{\pma}[2]			{{{\mP}_{{#1}}^{{#2}}}}

\newcommand{\qrv}[0]			{{\rndv{Q}}}

\newcommand{\mQ}[0]				{{\msr{q}}}    
\newcommand{\qmn}[1]			{{{\mQ}_{{#1}}}}

\newcommand{\mU}[0]				{{\msr{u}}}

\newcommand{\Um}[0]				{{{\cnst{U}}}}

\newcommand{\vrv}[0]			{{\rndv{V}}}

\newcommand{\mV}[0]				{{\msr{v}}}

\newcommand{\mW}[0]				{{\msr{w}}}    

\newcommand{\wma}[2]			{{{\mW}_{{#1}}^{{#2}}}}
\newcommand{\Wm}[0]				{{{\cnst{W}}}}
\newcommand{\Wmn}[1]			{{{\cnst{W}}_{{#1}}}}
\newcommand{\Wma}[2]			{{{\cnst{W}}_{{#1}}^{{#2}}}}

\newcommand{\smplS}[0]				{{{\Omega}}}
\newcommand{\dsmpl}[0]				{{\cnst{\omega}}}
\newcommand{\smplA}[0]				{{\alg{F}}}
\newcommand{\fltrn}[1]				{{\alg{F}}_{{#1}}}
\newcommand{\smplG}[0]				{{\alg{G}}}

\makeatletter
\DeclareRobustCommand{\bigplus}{%
	\mathop{\vphantom{\sum}\mathpalette\@bigplus\relax}\slimits@
}
\newcommand{\@bigplus}[2]{\vcenter{\hbox{\make@bigplus{#1}}}}
\newcommand{\make@bigplus}[1]{%
	\sbox\z@{$\m@th#1\sum$}%
	\setlength{\unitlength}{\wd\z@}%
	\begin{picture}(1.4,1.4)
	\linethickness{.17ex}
	\Line(.7,.14)(.7,1.26)
	\Line(.14,.7)(1.26,.7)
	\end{picture}%
}
\DeclareRobustCommand{\bigtimes}{%
	\mathop{\vphantom{\sum}\mathpalette\@bigtimes\relax}\slimits@
}
\newcommand{\@bigtimes}[2]{\vcenter{\hbox{\make@bigtimes{#1}}}}
\newcommand{\make@bigtimes}[1]{%
	\sbox\z@{$\m@th#1\sum$}%
	\setlength{\unitlength}{\wd\z@}%
	\begin{picture}(1,1)
	\linethickness{.17ex}
	\Line(.1,.1)(.9,.9)
	\Line(.1,.9)(.9,.1)
	\end{picture}%
}
\makeatother

\newcommand{\prob}[0]        {{\bf P}}
\renewcommand{\PXS}[2]       {{\prob}_{\!{#1}}\!\left({#2}\right)}
\renewcommand{\PX}[1]        {\PXS{}{{#1}}}                      
\renewcommand{\PCXS}[3]      {\PXS{{#1}}{\left.\! {#2} \right| {#3}}}
\renewcommand{\PCX}[2]       {\PCXS{}{{#1}}{{#2}}}  
\newcommand{\TCIND}[2]       {{\mathds{1}_{#1}\left({#2}\right)}}  

\begin{document}
\title{The Sphere Packing Bound\\ for DSPCs with Feedback \`{a} la Augustin}
\author{Bar\i\c{s} Nakibo\u{g}lu\vspace{-.3cm}
\thanks{e-mail:bnakib@metu.edu.tr
%
%
}
}
\maketitle
\pagestyle{plain}
\pagenumbering{arabic}
\thispagestyle{empty} 
\begin{abstract}
Establishing the sphere packing bound for block codes on 
the discrete stationary product channels with feedback 
---which are commonly called the discrete memoryless channels with feedback---
was considered to be an open problem until recently, 
notwithstanding the proof sketch provided by Augustin in 1978.
A complete proof following Augustin's proof sketch is presented,
to demonstrate its adequacy and to draw attention to 
two novel ideas it employs. These novel ideas  
(i.e., the Augustin's averaging and the use of subblocks) 
are likely to be applicable in other communication problems 
for establishing impossibility results.    
\end{abstract}
\section{Introduction}\label{sec:introduction}
After the founding paper of Shannon \cite{shannon48}, 
establishing the channel capacity as the threshold rate 
for reliable communication, 
one of the first challenges
of the mathematical theory of communications
was determining the behavior of 
the optimum error probability as a function of the block length
at rates below the channel capacity.
The optimum error probability was shown to decay exponentially 
with the block length 
and the exponent of this decay (i.e., the error exponent, or the reliability function) 
was determined at all rates between 
the critical rate and the capacity of the channel
in \cite{shannon59,elias55B,dobrushin62B} 
for various channel models.
Although it was not always discussed in these terms, 
\cite{shannon59,elias55B,dobrushin62B} proved the following 
two distinct results in order determine the
error exponent at rates between the critical rate and the capacity
of the channel.
\begin{enumerate}[(i)]
\item\label{RCE} {\it The Random Coding Bound (RCB):} 
At all rates less than the capacity, 
the random coding exponent (RCE) is achievable,
i.e., the error exponent is bounded from below by the RCE.
\item\label{SPB} {\it The Sphere Packing Bound (SPB):}
At any rate less than the capacity the error exponent 
is bounded from above by the sphere packing exponent (SPE).
\end{enumerate}
The RCE and the SPE are equal 
to one another for all rates between the critical rate and
the channel capacity.
Thus \eqref{RCE} and \eqref{SPB} determine the error exponent 
exactly for all rates between the critical rate and the capacity
on any channel that they are established.

In \cite{gallager65}, Gallager proved \eqref{RCE} not only for
all of the models considered in \cite{shannon59,elias55B,dobrushin62B},
but also for essentially all memoryless channel models of interest, 
including the non-stationary ones. 
The elegance and the simplicity of Gallager's derivation and 
the generality of his result make his seminal paper \cite{gallager65}
of interest to the contemporary researchers after decades
\cite{shinW09,zhangMKWT13}. 

For the SPB ---i.e., for \eqref{SPB}--- the progress did not happen all at once
as it did for \eqref{RCE}.
The first two complete proofs of the SPB for arbitrary discrete stationary 
product channels 
(DSPCs)\footnote{The channels that we call DSPCs are usually called 
	discrete memoryless channels
	(i.e., DMCs).
	We use the name DSPC to underline the stationarity of these channels 
	and the non-existence of constraints on their input sets;
	see \S \ref{sec:model} for a more detailed discussion.} 
by Shannon, Gallager, and Berlekamp in \cite{shannonGB67A}
and by Haroutunian in \cite{haroutunian68} both relied
on expurgations based on the composition (i.e., the empirical distribution, 
or the type)
of the input codewords. 
Thus the proofs in \cite{shannonGB67A} and \cite{haroutunian68} 
hold only for codes on stationary channels with finite input sets.
In \cite{augustin69},
Augustin provided the first proof of the SPB on the product channels that 
does not assume either the stationarity of the channel or 
the finiteness of its input set.
In \cite{nakiboglu19B}, we have improved 
the approximation error term of the upper bound on the error exponent 
given in \cite{augustin69}
from \(\bigo{\blx^{-0.5}}\) to \(\bigo{\blx^{-1} \ln \blx}\)
for the block length \(\blx\), using the R\'enyi capacity
and center analyzed in \cite{nakiboglu19A}.

Unlike the proofs in \cite{shannonGB67A} and \cite{augustin69},
Haroutunian's proof in \cite{haroutunian68} establishes 
the SPB not only for codes on the product channels but also 
for codes on the stationary memoryless channels
with either composition or cost constraints.
However, the finite input set hypothesis of \cite{haroutunian68}   
curbs its usefulness for models other than the discrete ones, 
e.g., \cite{haroutunian68} does not imply the SPB 
for the Poisson channels, derived for the first time in 
\cite{wyner88-b}.
Building upon the techniques he 
developed in \cite{augustin69}  and \cite[\S31]{augustin78}
and employing the information measures he analyzed in \cite[\S34]{augustin78},
Augustin proved the SPB on (possibly non-stationary) cost
constrained memoryless channels 
with bounded cost functions in \cite[\S36]{augustin78}.
Augustin's SPB given in \cite[Thm. 36.6]{augustin78} 
applies to the Poisson channels,
but not to various Gaussian channels 
analyzed in \cite{shannon59,ebert66,richters67}
because the quadratic cost function is not bounded.
In \cite{nakiboglu18D}, we have proved the SPB 
for codes on the cost constrained memoryless channels
---without assuming the cost function to be bounded---
using the constrained Augustin capacity and center 
analyzed in \cite[\S34]{augustin78} and \cite{nakiboglu18C}.
The SPB given in 
\cite[Thm. 2]{nakiboglu18D} 
implies the SPB not only for the Poisson channels,
but also for various Gaussian channels 
considered in \cite{shannon59,ebert66,richters67}.

Despite their generality, Augustin proofs in \cite{augustin69} and \cite{augustin78} 
did not have nearly as much impact as the proofs in \cite{shannonGB67A} and \cite{haroutunian68}.
This is partly due to the considerable simplification 
provided by the application of the composition based expurgations
in \cite{shannonGB67A} and \cite{haroutunian68}.
This reliance on the composition based expurgations, however, 
were making the derivation of the SPB with the techniques 
in \cite{shannonGB67A} and \cite{haroutunian68} 
rather convoluted and tedious ---if at all possible---
for codes on channels other than the
stationary memoryless ones with finite input sets.
For codes on DSPCs with feedback, for example, there is no evident generalization 
for the concept of composition of an input codeword that can be used 
in a derivation of the SPB similar to  \cite{shannonGB67A} or \cite{haroutunian68}. 
Thus establishing the SPB for arbitrary DSPCs with feedback has been 
a significant challenge. 
Nevertheless, several partial results have been reported over the years.

For DSPCs with feedback that have certain symmetries,
Dobrushin established the SPB in \cite{dobrushin62A}.
For arbitrary DSPCs with feedback,
Haroutunian  \cite{haroutunian77} derived an upper bound
on the error exponent,
which is usually called Haroutunian's bound/exponent.
Haroutunian's exponent is equal to the SPE only for
DSPCs with certain symmetries;
Haroutunian's exponent is strictly greater than the SPE
even for non-symmetric binary input binary output channels.
Sheverdyaev proposed a derivation of the SPB for codes on
DSPCs with feedback using Taylor's expansion
in \cite{sheverdyaev82}. 
Sheverdyaev's proof was, however, supported rather weakly on several critical
steps, see \cite[A7]{palaiyanurthesis} for a more detailed discussion.
Curtailing the ways feedback link can be used
by appropriate assumptions,
\cite{palaiyanurthesis,palaiyanurS10,comoN10} derived the SPB 
for certain families of codes on arbitrary DSPCs with feedback. 

Augustin presented a proof sketch establishing the SPB 
on arbitrary DSPCs with feedback in \cite[\S 41]{augustin78}.
Despite the novelty of Augustin's approach and the importance 
of his result, Augustin's proof sketch is not widely known.
In fact, until very recently, establishing the SPB on DSPCs with 
feedback has been considered to be an open problem.
In the following, we present a complete 
proof that is following Augustin's proof sketch without any 
significant modification.
Our main aim is to make the two main ideas of Augustin's proof
---the averaging and the use of subblocks--- widely accessible
via this relatively short article.
We believe both ideas are likely to be useful in establishing 
impossibility results in other communications problems.
We assume the channel to be discrete for simplicity and
employ concepts that are not present, at least explicitly,
in \cite{augustin78}
---such as R\'enyi's information measures and stochastic sequences--- 
whenever we think  their use simplifies the discussion
for the contemporary researcher.

Elsewhere, in \cite[\S V]{nakiboglu19B}, 
we have proved the SPB for codes on DSPCs with feedback 
using the averaging and the subblock ideas of Augustin \cite{augustin78}
together with the Taylor's expansion idea of Sheverdyaev \cite{sheverdyaev82}
and the auxiliary channel idea of Haroutunian \cite{haroutunian68,haroutunian77}.
In addition, we have shown in \cite[\S V-E]{nakiboglu19B} that
Haroutunian's bound implies the SPB when considered together with 
the averaging and the subblock ideas of Augustin.
Although proofs in \cite[\S V]{nakiboglu19B} do employ ideas from 
Augustin's proof sketch, both proofs also employ 
other fundamental observations which makes them substantially 
different from the proof we present in the following.

In the rest of the current section, 
we first describe our notation and model, 
then state the main asymptotic result, i.e., 
Theorem \ref{proposition:SPBF}.
In \S\ref{sec:preliminaries}, we recall certain properties of 
R\'enyi's information measures and SPE,
derive preliminary results on tilting and stochastic sequences,
and state a sufficient condition for constructing a 
probability measure with a given set of conditional probabilities
on a product space.
In \S\ref{sec:SPB}, we prove a non-asymptotic SPB 
for codes on DSPCs with feedback,
which implies Theorem \ref{proposition:SPBF}.
In \S\ref{sec:conclusion}, we discuss 
possible generalizations and alternative proofs
for the main result of the paper, establishing the
sphere packing exponent as an upper bound to the reliability function
for channel with feedback.

\subsection{Notation}\label{sec:notation}
We denote the set of all real numbers by \(\reals{}\), 
positive real numbers by \(\reals{+}\), 
non-negative real numbers by \(\reals{\geq0}\),
and integers by \(\integers{\!}\). 
For any  real number \(\dsta\), 
\(\lfloor\dsta\rfloor\) is the greatest integer less than or equal to \(\dsta\),
\(\lceil\dsta\rceil\) is the least integer greater than or equal to  \(\dsta\),
and \(\abs{\dsta}\) is the absolute value of \(\dsta\).
For any set \(\cset\) the indicator function 
\(\TCIND{\cset}{\cdot}\) is defined as follows:
\begin{align}
\notag
\TCIND{\cset}{\dinp}
&=
\begin{cases}
1
&\dinp\in\cset
\\
0
&\dinp\notin\cset
\end{cases}.
\end{align}

For any finite set \(\outS\), we denote 
the set of all subsets of \(\outS\) 
(i.e., the power set of \(\outS\)) by \(\sss{\outS}\)
and 
the set of all probability mass functions (\pmf\!\!\!'s) on \(\outS\) by \(\pdis{\outS}\).  
For any \(\mQ\) and \(\mW\) in \(\pdis{\outS}\) 
the total variation distance between them is defined as
\begin{align}
\lon{\mQ-\mW}
\DEF \sum\nolimits_{\dout \in \outS} \abs{\mQ(\dout)-\mW(\dout)}.
\end{align}

While discussing the continuity of functions, we will assume that 
the set of real numbers is equipped with its natural topology
and the set of all \pmf\!\!\!'s is equipped with the total variation 
topology.

For any two finite sets \(\inpS\) and \(\outS\), we denote 
the Cartesian product of  \(\inpS\) and \(\outS\) by \(\inpS\times\outS\),
the set of all functions from \(\inpS\) to \(\outS\) by \({\outS}^{\inpS}\),
and
the set of all stochastic matrices from 
\(\inpS\) to \(\outS\) by \(\pdis{\outS|\inpS}\).
We interpret stochastic matrices from \(\inpS\) to \(\outS\)
as functions from \(\inpS\) to \(\pdis{\outS}\), as well.
Thus we use \(\Wm(\dinp)\) and \(\Wm(\cdot|\dinp)\)
interchangeably for \(\Wm\)'s in \(\pdis{\outS|\inpS}\). 
For any \(\mP\) in \(\pdis{\inpS}\) and \(\Wm\) in \(\pdis{\outS|\inpS}\),
\(\mP\mtimes\Wm\) is the \pmf on \(\inpS\times\outS\) whose marginal 
distribution on
\(\inpS\) is \(\mP\) and conditional distribution given \(\dinp\) is
\(\Wm(\dinp\!)\). For  any
\(\mP\) in \(\pdis{\inpS}\) and \(\mQ\) in \(\pdis{\outS}\), 
we denote their product, which is a 
\pmf on \(\inpS\!\times\!\outS\), by \(\mP\!\otimes\!\mQ\).
We use the symbol \(\otimes\) to denote 
the product of \(\sigma\)-algebras, as well.

For any interval \(\cset\) on \(\reals{}\)
the Borel \(\sigma\)-algebra of \(\cset\),
denoted by \(\rborel{\cset}\), 
is the minimum \(\sigma\)-algebra on the subsets of 
\(\cset\) that includes all the open subintervals of
\(\cset\), \cite[p. 143]{shiryaev}.
A pair \((\smplS,\smplA)\) is a measurable space iff 
\(\smplA\) is a \(\sigma\)-algebra of subsets of \(\smplS\).
If in addition \(\prob\) is a probability on \(\smplA\), then 
the triple  \((\smplS,\smplA,\prob)\) form a probability space.
A real valued function \(\inp\) on \(\smplS\) is a random variable 
in the probability space \((\smplS,\smplA,\prob)\)
iff \(\inp\) is \(\smplA\)-measurable
(i.e., the inverse image of every set in \(\rborel{\reals{}}\) is in \(\smplA\)),
\cite[p. 170]{shiryaev}.
A sequence of pairs
\((\inp_{1},\fltrn{1}),\ldots,(\inp_{\blx},\fltrn{\blx})\)
is a stochastic sequence 
in \((\smplS,\smplA,\prob)\)
iff 
\(\fltrn{1},\ldots,\fltrn{\blx}\) are \(\sigma\)-algebras
satisfying \(\fltrn{1}\subset\cdots\fltrn{\blx}\subset\smplA\)
and
\(\inp_{\tin}\)'s are \(\fltrn{\tin}\)-measurable random variables,
\cite[p. 476]{shiryaev}.
See \cite[Ch. II]{shiryaev}, for an  accessible introduction to
the mathematical foundations of the probability theory.

Our notation will be overloaded for certain symbols,
but the relations represented by these symbols
will be clear from the context.
We use the short hand
\(\smplG_{\tin}^{\blx}\) for the product of \(\sigma\)-algebras \(\smplG_{\tin},\ldots,\smplG_{\blx}\),
\(\inpS_{\tin}^{\blx}\) for the Cartesian product of sets \(\inpS_{\tin},\ldots,\inpS_{\blx}\),
\(\inp_{\tin}^{\blx}\) for the random vector \((\inp_{\tin},\ldots,\inp_{\blx})\),
and
\(\dinp_{\tin}^{\blx}\) for the vector \((\dinp_{\tin},\ldots,\dinp_{\blx})\).

\subsection{The DSPCs with Feedback and the Channel Codes}\label{sec:model}
A discrete channel with 
a finite input set \(\inpS\) and a finite output set \(\outS\),
is represented by a stochastic matrix \(\Wm\).
The product of a sequence of discrete channels 
\(\Wmn{1},\ldots,\Wmn{\blx}\)
with the input sets \(\inpS_{1},\ldots,\inpS_{\blx}\)
and the output sets \(\outS_{1},\ldots,\outS_{\blx}\)
is a discrete channel from 
\(\inpS_{1}^{\blx}\) to \(\outS_{1}^{\blx}\),
denoted by \(\Wmn{[1,\blx]}\), satisfying
\begin{align}
\notag
\Wmn{[1,\blx]}(\dout_{1}^{\blx}|\dinp_{1}^{\blx})
&=\prod\nolimits_{\tin=1}^{\blx} \Wmn{\tin}(\dout_{\tin}|\dinp_{\tin})
\end{align}
for all \(\dinp_{1}^{\blx}\) in \(\inpS_{1}^{\blx}\)
and \(\dout_{1}^{\blx}\) in \(\outS_{1}^{\blx}\).
A length \(\blx\) product channel \(\Wmn{[1,\blx]}\) is stationary iff 
all \(\Wmn{\tin}\)'s are identical. 
A discrete channel \(\Um\) from \(\staS\) to \(\outS_{1}^{\blx}\)
is a length \(\blx\) memoryless channel if there exits a product channel 
\(\Wmn{[1,\blx]}\) with the input set \(\inpS_{1}^{\blx}\)
satisfying both \(\staS\subset\inpS_{1}^{\blx}\)
and 
\(\Um(\dsta)=\Wmn{[1,\blx]}(\dsta)\)
for all \(\dsta\in\staS\). 

The preceding definition of the memorylessness is wholly consistent 
with the one used in standard texts \cite[p. 185]{coverthomas}, 
\cite[(4.2.1)]{gallager}, \cite[p. 84]{csiszarkorner}.
Nevertheless, the discrete product channels that are
also stationary 
are customarily called discrete memoryless channels.
Although the conventional name is not wrong, 
we prefer a more descriptive and accurate name: 
the discrete stationary product channels (DSPCs).

In discrete product channels (DPCs) probabilistic behavior of the channel outputs 
depend on the channel inputs, but the channel inputs do not depend on the channel 
outputs in any way. 
In DPCs with feedback, on the other hand,
the channel input at any time instance may depend 
on the previous channel outputs,
i.e., the channel input at time \(\tin\) can be a function 
from  \(\outS_{1}^{\tin-1}\) to \({\inpS_{\tin}}\) 
rather than an element of  \(\inpS_{\tin}\).
We define the DPCs with feedback formally as follows.

\begin{definition}\label{def:Fproduct:discrete}
	For any positive integer \(\blx\) and 
	\(\Wmn{\tin}\!:\!\inpS_{\tin}\!\to\!\pdis{\outS_{\tin}}\) for \(\tin\) 
	in \(\{1,\ldots,\blx\}\),
	the \emph{length \(\blx\) discrete product channel with feedback}
	\(\Wmn{\vec{[1,\blx]}}\!:\!\vec{\inpS}_{1}^{\blx}\!\to\!\pdis{\outS_{1}^{\blx}}\) 
	is defined via the following relation:
	\begin{align}
	\label{eq:def:Fproduct:discrete}
	\Wmn{\vec{[1,\blx]}}(\dout_{1}^{\blx}|\vec{\dinp}_{1}^{\blx}) 
	&\!=\!
	\Wmn{1}(\dout_{1}|\vec{\dinp\!}_{\!1})
	\prod\nolimits_{\tin=2}^{\blx}
	\Wmn{\tin}(\dout_{\tin}|\vec{\dinp\!}_{\!\tin}(\dout_{1}^{\tin-1})) 
	\end{align}
	for all \(\vec{\dinp}_{1}^{\blx}\!\in\!\vec{\inpS}_{1}^{\blx}\) 
	and \(\dout_{1}^{\blx}\!\in\!\outS_{1}^{\blx}\)
	where \(\vec{\inpS}_{\tin}\!=\!{\inpS_{\tin}}^{\outS_{1}^{\tin-1}}\)
	for \(\tin\geq 2\)
	and \(\vec{\inpS}_{1}\!=\!\inpS_{1}\).
A DPC with feedback \(\Wmn{\vec{[1,\blx]}}\) is stationary, 
i.e., it is a DSPC with feedback, iff all \(\Wmn{\tin}\)'s are identical. 
\end{definition}

Broadly speaking, 
a channel code is a strategy to convey from the transmitter 
at the input of the channel to the receiver at the output 
of the channel, a random choice from a finite message set. 
The channel codes are usually described in terms of
the amount of information they convey  
per channel use, i.e., in terms of their rate. 
In particular, a rate \(\rate\)  \emph{channel code} on 
a length \(\blx\) DPC with feedback \(\Wmn{\vec{[1,\blx]}}\) 
is an ordered pair \((\enc,\dec)\) composed 
of the \emph{encoding function}  \(\enc\) 
that maps the message set 
\(\mesS\DEF\{1,2,\ldots,\lceil e^{\blx \rate}\rceil\}\) 
to the input set \(\vec{\inpS}_{1}^{\blx}\)
and the \emph{decoding function} \(\dec\)
that maps the output set \(\outS_{1}^{\blx}\) 
to the message set \(\mesS\). 

\emph{The average error probability} \(\Pem{av}\) of
a rate \(\rate\) channel code \((\enc,\dec)\) 
on a length \(\blx\) DPC with feedback \(\Wmn{\vec{[1,\blx]}}\)
is
\begin{align}
\label{eq:def:errorprobability}
\Pem{av} 
&\DEF\tfrac{1}{\lceil e^{\blx \rate}\rceil} \sum\nolimits_{\dmes\in \mesS} \Pem{\dmes},
\end{align}
where \(\Pem{\dmes}\) is \emph{the conditional error probability} 
of the message \(\dmes\)  given by 
\begin{align}
\label{eq:def:conditionalerrorprobability}
\Pem{\dmes}
&\DEF 1-\sum\nolimits_{\dout_{1}^{\blx}\in \outS_{1}^{\blx}}\!
\TCIND{\{\dec(\dout_{1}^{\blx})\}}{\dmes}
\Wmn{\vec{[1,\blx]}}(\dout_{1}^{\blx}|\enc(\dmes)).
\end{align}

\subsection{Main Result}\label{sec:mainresult}
\begin{definition}\label{def:information}
	For any \(\rno\!\in\!(0,1]\), \(\Wm\!\in\!\pdis{\outS|\inpS}\), 
	and \(\mP\!\in\!\pdis{\inpS}\)
	\emph{the order-\(\rno\) R\'enyi information for prior \(\mP\)} is 
	\begin{align}
	\notag
	\RMI{\rno}{\mP}{\Wm} 
	&\!\DEF\!\!
	\begin{cases}
	\!\tfrac{\rno}{\rno-1}\ln\!\sum\limits_{\dout}
	\left[\sum\limits_{\dinp}\mP(\dinp)\left[\Wm(\dout|\dinp)\right]^{\rno}\right]^{\sfrac{1}{\rno}}
	&\rno\!\in\!(0,1)
	\\
	\sum\limits_{\dinp}\mP(\dinp)
	\sum\limits_{\dout}\Wm(\dout|\dinp)\ln\tfrac{\Wm(\dout|\dinp)}{\qmn{1,\mP}(\dout)}
	& \rno\!=\!1
	\end{cases},
	\end{align} 
	where \(\qmn{1,\mP}\in\pdis{\outS}\) is defined as \(\qmn{1,\mP}(\dout)\!\DEF\!\sum_{\dinp}\mP(\dinp)\Wm(\dout|\dinp)\).
\end{definition}
\begin{definition}\label{def:capacity}
	For any \(\rno\!\in\!(0,1]\)  and \(\Wm\!\in\!\pdis{\outS|\inpS}\) 
	\emph{the order-\(\rno\) R\'enyi capacity of \(\Wm\)} is
	\begin{align}
	\notag
	\RC{\rno}{\Wm} 
	&\DEF \sup\nolimits_{\mP \in \pdis{\inpS}}  \RMI{\rno}{\mP}{\Wm}.
	\end{align}
\end{definition}
Both the R\'enyi information and the R\'enyi capacity are 
continuous non-decreasing functions of the order \(\rno\) 
on \((0,1]\), 
see \cite[Lemmas 5 and 15]{nakiboglu19A}.  
We define the order-\(0\) R\'enyi capacity as the continuous extension 
of the R\'enyi capacity at zero:\footnote{The order-\(0\) R\'enyi information is defined in a similar way 
	and the supremum  \(\RMI{0}{\mP}{\Wm}\)
	over \(\mP\)'s in \(\pdis{\inpS}\) is equal to 
	\(\RC{0}{\Wm}\), as defined in \eqref{eq:def:capacity-zero},
	see \cite[Lemma 16-(f)]{nakiboglu19A}.
}
\begin{align}
\label{eq:def:capacity-zero}
\RC{0}{\Wm}
\DEF \lim\nolimits_{\rno\downarrow 0}\RC{\rno}{\Wm}.
\end{align}
\begin{definition}\label{def:spherepacking}
	For any stochastic matrix \(\Wm\!\in\!\pdis{\outS|\inpS}\) and  
	rate \(\rate\!\in\!\reals{\geq0}\), 
	\emph{the sphere packing exponent} (SPE) is
	\begin{align}
	\notag
	\spe{\rate,\Wm}
	&\DEF \sup\nolimits_{\rno\in (0,1)} \tfrac{1-\rno}{\rno} \left(\RC{\rno}{\Wm}-\rate\right).
	\end{align} 
\end{definition}

Note that if \(\RC{0}{\Wm}\!=\!\RC{1}{\Wm}\),
then \(\spe{\rate,\Wm}\) is 
infinite for \(\rate\)'s in \([0,\RC{1}{\Wm})\)
and
zero for \(\rate\)'s in \([\RC{1}{\Wm},\infty)\).
For most stochastic matrices of interest, however,
  \(\RC{1}{\Wm}\!>\!\RC{0}{\Wm}\)
and consequently \(\spe{\rate,\Wm}\) is 
a convex function of \(\rate\) that is 
infinite on 
\([0,\RC{0}{\Wm})\),
monotonically decreasing and continuous in \(\rate\) 
on \((\RC{0}{\Wm},\RC{1}{\Wm}]\),
and
zero on \([\RC{1}{\Wm},\infty)\), see 
\cite[Lemma 13]{nakiboglu19B}.
\begin{remark}\label{rem:expressions}
For orders in \((0,1)\) the R\'enyi information is just a scaled and 
reparameterized version of the Gallager's function
\(E_{0}(\rng,\mP)\) introduced in \cite{gallager65}; in particular
\begin{align}
\notag
\RMI{\rno}{\mP}{\Wm}
&=\left.\tfrac{E_{0}(\rng,\mP)}{\rng}\right\vert_{\rng=\frac{1-\rno}{\rno}}
&
&\forall \rno\in (0,1).
\end{align}
In \cite{shannonGB67A}, 
the function \(E_{0}(\rng)\) is defined as the maximum of Gallager's 
function \(E_{0}(\rng,\mP)\) over \(\mP\)'s in \(\pdis{\inpS}\).
Thus
\begin{align}
\notag
\RC{\rno}{\Wm}
&=\left.\tfrac{E_{0}(\rng)}{\rng}\right\vert_{\rng=\frac{1-\rno}{\rno}}
&
&\forall \rno\in (0,1).
\end{align}
Consequently, Definition \ref{def:spherepacking} 
is merely a reparameterization of the definition used
by Shannon, Gallager, and Berlekamp in \cite[Thm. 2]{shannonGB67A}.
In \cite{haroutunian68}, Haroutunian employed another expression for the
SPE, which he proved to be equal to the one
in \cite{shannonGB67A}.
This expression is commonly known as Haroutunian's form.
\end{remark}

\begin{theorem}\label{proposition:SPBF}
	For any \(\Wm\!\in\!\pdis{\outS|\inpS}\) satisfying \(\RC{0}{\Wm}\!\neq\!\RC{1}{\Wm}\),
	and \(\rate_{0}\), \(\rate_{1}\) satisfying
	\(\RC{0}{\Wm}<\rate_{0}<\rate_{1}<\RC{1}{\Wm}\), 
	for all \(\blx\) large enough 
	\begin{align}
	\label{eq:thm:SPBF}
	\Pem{av}
	&\geq \exp \left(-\blx \left[\spe{\rate-\tfrac{2\ln \blx}{\blx^{\sfrac{1}{3}}},\Wm}+\tfrac{2\ln \blx}{\blx^{\sfrac{1}{3}}}\right] \right)	 
	\end{align}
	for any rate \(\rate\) channel code on the length \(\blx\) DSPC 
	with feedback \(\Wmn{\vec{[1,\blx]}}\)  satisfying
	\(\Wmn{\tin}=\Wm\) 
	provided \(\rate\) satisfies 
	\begin{align}
	\label{eq:thm:SPBF-hypothesis}
	\rate_{1}>\rate 	
	&>\rate_{0}+\tfrac{2\ln \blx}{\blx^{\sfrac{1}{3}}}.
	\end{align}
\end{theorem}

Note that 
\(\tfrac{2\ln \blx}{\blx^{\sfrac{1}{3}}}\) terms 
in \eqref{eq:thm:SPBF} and \eqref{eq:thm:SPBF-hypothesis}
vanish as \(\blx\)  increases; 
thus Theorem \ref{proposition:SPBF} establishes
the SPE as an upper bound on the error exponent of 
any DSPC with feedback 
at any rate in \((\!\RC{0}{\Wm},\!\RC{1}{\Wm}\!)\),
provided that \(\Wmn{\tin}\!=\!\Wm\) for all \(\tin\).
In fact this result holds with uniform approximation error 
terms on every closed interval of rates in 
\((\!\RC{0}{\Wm},\!\RC{1}{\Wm}\!)\),
as a result of Theorem \ref{proposition:SPBF}.
For rates less than \(\RC{0}{\Wm}\), SPE is infinite; thus 
the upper bound holds trivially.
For rates larger than \(\RC{1}{\Wm}\), we already know that 
the optimal error probability 
of the channel codes
converges to one by \cite{csiszarK82,sheverdyaev82}.

\section{Preliminaries}\label{sec:preliminaries}
\subsection{R\'enyi's Information Measures and SPE}\label{sec:renyi}
R\'enyi's information measures have been studied 
explicitly \cite{renyi61,sibson69,csiszar95}
or implicitly \cite{gallager65,shannonGB67A} 
since the sixties. 
For the finite sample space case, 
the propositions about them that we borrow from 
\cite{nakiboglu19A} and \cite{ervenH14} in the following 
are relatively easy to prove
and well-known, except for Lemma \ref{lem:centercontinuity} establishing 
the continuity of the R\'enyi center as a function of the order.
Lemma \ref{lem:intermediate} states an immediate corollary of 
the monotonicity properties of the R\'enyi capacity and the definition
of the SPE.

\begin{definition}\label{def:divergence}
	For any \(\rno\!\in\!(0,1]\) and \(\mW,\mQ\!\in\!\pdis{\outS}\),
	\emph{the order-\(\rno\) R\'enyi divergence between \(\mW\) and \(\mQ\)} is 
	\begin{align}
	\notag
	\RD{\rno}{\mW}{\mQ}
	&\!\DEF\! 
	\begin{cases}
	\sum\nolimits_{\dout} \mW(\dout)\ln \tfrac{\mW(\dout)}{\mQ(\dout)}
	&\rno\!=\!1
	\\
	\tfrac{1}{\rno-1} \ln \sum\nolimits_{\dout}\left[\mW(\dout)\right]^{\rno} \left[\mQ(\dout)\right]^{1-\rno}
	&\rno\!\neq\!1 
	\end{cases}.
	\end{align} 
\end{definition}
Note that for all \(\rno\!\in\!(0,1)\) and \(\mW,\mQ\!\in\!\pdis{\outS}\)
we have
\begin{align}
\label{eq:symmetry}
\tfrac{1-\rno}{\rno}\RD{\rno}{\mW}{\mQ}
&=\RD{1-\rno}{\mQ}{\mW}
\end{align}
by definition.
Using the derivatives of \(e^{(\rno-1)\RD{\rno}{\mW}{\mQ}}\)
with respect to \(\rno\),
one can show that as a function of its order the R\'enyi divergence 
is nondecreasing on \((0,1)\) and  continuous from the left at one. 
Thus, we get the following proposition.
\begin{lemma}[\!\!{\cite[Thms. 3, 7]{ervenH14}}]\label{lem:divergence}
For any \(\mW,\mQ\!\in\!\pdis{\outS}\), the R\'enyi divergence 
	\(\RD{\rno}{\mW}{\mQ}\) is nondecreasing 
	and continuous in \(\rno\) on \((0,1]\).
\end{lemma}

The R\'enyi divergence is non-negative 
as a result of the Jensen's inequality.
This observation has been strengthened 
by the following inequality relating 
the R\'enyi divergence to the total variation distance 
\cite{csiszar67A}, \cite{gilardoni10B},
called the Pinsker's inequality:
\begin{align}
\label{eq:pinsker}
\RD{\rno}{\mW}{\mQ}
&\geq \tfrac{\rno}{2} \lon{\mW-\mQ}^{2} 
\end{align}
for all \(\rno\!\in\!(0,1]\) and \(\mW,\mQ\!\in\!\pdis{\outS}\).
\begin{definition}\label{def:radius}
For any \(\rno\!\in\!(0,1]\)  and \(\Wm\!\in\!\pdis{\outS|\inpS}\) 
\emph{the order-\(\rno\) R\'enyi radius of \(\Wm\)} is
	\begin{align}
	\notag
\RR{\rno}{\Wm}
&\!\DEF\!\inf\nolimits_{\mQ\in\pdis{\outS}}\max\nolimits_{\dinp \in \inpS} \RD{\rno}{\Wm(\dinp)}{\mQ}.
	\end{align}
\end{definition}
The order-\(\rno\) R\'enyi capacity is defined as the supremum of
the order-\(\rno\) R\'enyi information; however, it is also equal to 
the order-\(\rno\) R\'enyi radius, \cite[Proposition 1]{csiszar95}.
In addition, there exists a unique order-\(\rno\) R\'enyi center
corresponding to this radius.
These observations are stated formally in Lemma \ref{thm:minimax}.
\begin{lemma}[\!\!{\cite[Thm. 1]{nakiboglu19A}}]
\label{thm:minimax}
For any \(\rno\!\in\!(0,1]\)  and \(\Wm\!\in\!\pdis{\outS|\inpS}\)
	\begin{align}
	\label{eq:thm:minimaxradius}
	\RC{\rno}{\Wm}
	&\!=\!\inf\nolimits_{\mQ\in\pdis{\outS}}\max\nolimits_{\dinp \in \inpS} \RD{\rno}{\Wm(\dinp)}{\mQ}.
	\end{align}
Furthermore, there exists a unique \(\qmn{\rno,\Wm}\) in \(\pdis{\outS}\),
	called the order-\(\rno\) R\'enyi center of \(\Wm\!\), such that
	\begin{align}
	\label{eq:thm:minimaxradiuscenter}
	\RC{\rno}{\Wm}
	&\!=\!\max\nolimits_{\dinp \in \inpS} \RD{\rno}{\Wm(\dinp)}{\qmn{\rno,\Wm}}.
	\end{align}
\end{lemma}
The R\'enyi capacity is nondecreasing in its order on \((0,1]\)
as a result of Lemmas \ref{lem:divergence} and \ref{thm:minimax}.
Furthermore, \(\tfrac{1-\rno}{\rno}\!\RC{\rno}{\Wm\!}\) is nonincreasing 
in \(\rno\) on \((0,1)\), as a result of  \eqref{eq:symmetry}
and  Lemmas \ref{lem:divergence} and \ref{thm:minimax}.
This implies the continuity of \(\RC{\rno}{\Wm}\) in \(\rno\) on \((0,1)\),
which can be extended to \((0,1]\).
\begin{lemma}[\!\!\!{\cite[Lemma\! 15-(a,c)]{nakiboglu19A}}]
\label{lem:capacityO}
For any \(\Wm\!\in\!\pdis{\outS|\inpS}\),
\(\RC{\rno}{\Wm}\) is nondecreasing and continuous in \(\rno\) on \((0,1]\)
and \(\tfrac{1-\rno}{\rno}\RC{\rno}{\Wm}\) is nonincreasing in \(\rno\) on \((0,1)\).
\end{lemma}
As a result of Lemma \ref{lem:capacityO}, we have
\begin{align}
\label{eq:orderoneovertwo}
\RC{\rno}{\Wm}
&\leq \tfrac{\RC{\sfrac{1}{2}}{\Wm}}{1-\rno}
&
&\forall \rno\in(0,1). 
\end{align}
The continuity of the R\'enyi capacity in the order implies the continuity of 
the R\'enyi center in the order.
\begin{lemma}[\!\!\!{\cite[Lemma 20]{nakiboglu19A}}]
\label{lem:centercontinuity} 
The R\'enyi center is a continuous function of its order  on \((0,1]\), i.e.,
\(\lim\nolimits_{\dsta\to \rno}\lon{\qmn{\dsta,\Wm}-\qmn{\rno,\Wm}}=0\) for all 
\(\rno\in(0,1]\).
\end{lemma}
The continuity of the R\'enyi center in the order allows us to construct 
a probability measure  that plays a crucial role in the proof of
Theorem \ref{proposition:SPBF}. 
\begin{proof}[Proof of Lemma \ref{lem:centercontinuity}]
	The following identity, which is due to Sibson \cite[p. 153]{sibson69},
	can be confirmed by substitution.
	\begin{align}
	\label{eq:sibson}
	\!\RD{\rno}{\mP \mtimes \Wm}{\!\mP\!\otimes\!\mQ} 
	&\!=\!
	\RMI{\rno}{\mP}{\Wm}\!+\!\RD{\rno}{\qmn{\rno,\mP}}{\mQ} 
	\end{align}
	where \(\qmn{\rno,\mP}\) is \emph{the order-\(\rno\) R\'enyi mean}
	defined as follows
	\begin{align}
	\label{eq:def:mean}
	\qmn{\rno,\mP}(\dout)
	&\DEF
	\tfrac{\left(\sum\nolimits_{\dinp}\mP(\dinp)\left[\Wm(\dout|\dinp)\right]^{\rno}\right)^{\sfrac{1}{\rno}}}{\sum\nolimits_{\mB} \left(\sum\nolimits_{\mA}\mP(\mA)\left[\Wm(\mB|\mA)\right]^{\rno}\right)^{\sfrac{1}{\rno}}}.
	\end{align}
	There exists a \(\pma{\rno}{*}\!\in\!\pdis{\inpS}\) such that \(\RMI{\rno}{\pma{\rno}{*}}{\Wm}\!=\!\RC{\rno}{\Wm}\)
	as a result of the extreme value theorem \cite[4.16]{rudin} because 
	\(\RMI{\rno}{\mP}{\Wm}\) is continuous in \(\mP\) on \(\pdis{\inpS}\) and 
	\(\pdis{\inpS}\) is compact.
	Note that \(\qmn{\rno,\pma{\rno}{*}}=\qmn{\rno,\Wm}\) by 
	\eqref{eq:pinsker},	\eqref{eq:sibson}, and Lemma \ref{thm:minimax}.
	Applying \eqref{eq:sibson} for \(\mQ=\qmn{\rnf,\Wm}\) and for \(\mP=\pma{\rno}{*}\) we get
	\begin{align}
	\notag
	\max\nolimits_{\dinp}\RD{\rno}{\Wm(\dinp)}{\qmn{\rnf,\Wm}}
	&\geq\RC{\rno}{\Wm}+\RD{\rno}{\qmn{\rno,\Wm}}{\qmn{\rnf,\Wm}}. 
	\end{align}
	Then using the monotonicity of R\'enyi divergence in the order
	(i.e., Lemma \ref{lem:divergence}) and
	Lemma \ref{thm:minimax} we get
	\begin{align}
	\notag
	\RC{\rnf}{\Wm}\!-\!\RC{\rno}{\Wm}
	&\!\geq\! 
	\RD{\rno}{\qmn{\rno,\Wm}}{\qmn{\rnf,\Wm}}
	&
	&\forall\rnf\!\in\![\rno,1].
	\end{align}
	Then the lemma
	follows from \eqref{eq:pinsker} and Lemma \ref{lem:capacityO}.
\end{proof}

\begin{lemma}\label{lem:intermediate}
For any stochastic matrix \(\Wm\!\in\!\pdis{\outS|\inpS}\) satisfying \(\RC{0}{\Wm}\!\neq\!\RC{1}{\Wm}\)
and rate \(\rate\) in \((\RC{0}{\Wm},\RC{1}{\Wm})\) there exists a \(\rnf\!\in\!(0,1)\)
satisfying \(\RC{\rnf}{\Wm}\!=\!\rate\) and an \(\rnt\!\in\!(\rnf,1)\)
satisfying \(\tfrac{1-\rnt}{\rnt}\RC{\rnt}{\Wm}\!=\!\spe{\rate,\!\Wm\!}\).
\end{lemma}

\begin{proof}[Proof of Lemma \ref{lem:intermediate}]
	Since \(\RC{\rno}{\Wm}\) is continuous in the order \(\rno\)
	by Lemma \ref{lem:capacityO}, the  existence of the order \(\rnf\)
	follows from the intermediate value theorem \cite[4.23]{rudin}.
Then, 
\begin{align}
\notag
\spe{\rate,\Wm}
&=\sup\nolimits_{\rno\in(\rnf,1)} \tfrac{1-\rno}{\rno} \left(\RC{\rno}{\Wm}-\rate\right),
\end{align}
because \(\RC{\rno}{\Wm}\) is non-decreasing  in the order \(\rno\) by Lemma \ref{lem:capacityO}. 
Thus \(\spe{\rate,\Wm}\) is positive 
at all rates \(\rate\) in \((\RC{0}{\Wm},\RC{1}{\Wm})\)  because 
\(\RC{\rnb}{\Wm}\!=\!\tfrac{\rate+\RC{1}{\Wm}}{2}\)
for some  \(\rnb\) in \((\rnf,1)\) by 
the intermediate value theorem \cite[4.23]{rudin}.
Then there exists an order \(\dsta\in[\rnf,1)\) satisfying
	\begin{align}
	\notag
	\spe{\rate,\Wm}
	&=\tfrac{1-\dsta}{\dsta} \left(\RC{\dsta}{\Wm}-\rate\right)
	\\
	\notag
	&<\tfrac{1-\dsta}{\dsta} \RC{\dsta}{\Wm}.
	\end{align}
Hence, \(\spe{\rate,\Wm}\) is between the values of the function
\(\tfrac{1-\rno}{\rno}\RC{\rno}{\Wm}\) at \(\rno\!=\!\dsta\)
and at \(\rno\!=\!1\).
Then the continuity of \(\tfrac{1-\rno}{\rno}\RC{\rno}{\Wm}\)
in the order \(\rno\) ---implied by Lemma \ref{lem:capacityO}---
and the intermediate value theorem \cite[4.23]{rudin}
imply the existence of the order \(\rnt\) in \((\dsta,1)\),
and hence in \((\rnf,1)\).
\end{proof}

\subsection{Tilting and the Selftilted Channel}\label{sec:tilting}
\begin{definition}\label{def:tiltedprobabilitymeasure}
	For any \(\rno\!\in\!(0,1]\) and \(\mW,\mQ\!\in\!\pdis{\outS}\)
	satisfying \(\RD{\rno}{\mW}{\mQ}\!<\!\infty\),
	\emph{the order-\(\rno\) tilted \pmf\!\!} 
	\(\wma{\rno}{\mQ}\) is 
	\begin{align}
	\notag
	\wma{\rno}{\mQ}(\dout)
	&\DEF e^{(1-\rno)\RD{\rno}{\mW}{\mQ}}
	[\mW(\dout)]^{\rno} [\mQ(\dout)]^{1-\rno}
	&
	&\forall\dout\in\outS.
	\end{align}
\end{definition}
One can confirm by substitution that
\begin{align}
\label{eq:tiltedKLD}
\rno\RD{1}{\wma{\rno}{\mQ}}{\mW}
\!+\!(1-\rno)\RD{1}{\wma{\rno}{\mQ}}{\mQ}
\!=\! (1-\rno)\RD{\rno}{\mW}{\mQ},
\end{align}
provided that \(\wma{\rno}{\mQ}\) is defined,
i.e., \(\RD{\rno}{\mW}{\mQ}\!<\!\infty\).

The continuity of the tilted \pmf \(\wma{\rno}{\mQ}\) in the order \(\rno\) on \((0,1)\)
is an immediate consequence of its definition and Lemma \ref{lem:divergence}. 
Interestingly, the continuity of \(\wma{\rno}{\mQ}\) in the order \(\rno\) on \((0,1)\)
holds even when \(\mQ\) is changing continuously with \(\rno\).
\begin{lemma}[\!\!{\cite[Lemma 16]{nakiboglu19B}}]
	\label{lem:tilting}
Let \(\qmn{\rno}\) be a continuous function of the order \(\rno\) from \((0,1)\) to \(\pdis{\outS}\)
and let \(\mW\!\in\!\pdis{\outS}\) satisfy \(\RD{\rno}{\mW}{\qmn{\rno}}\!<\!\infty\) for all \(\rno\in(0,1)\).
Then
\begin{enumerate}[(a)]
	\item\label{tilting-vma}
	  \(\wma{\rno}{\qmn{\rno}}\) is a continuous function of \(\rno\) from \((0,1)\)
to \(\pdis{\outS}\),
i.e., \(\lim\nolimits_{\dsta\to \rno}\lon{\wma{\dsta}{\qmn{\dsta}}-\wma{\rno}{\qmn{\rno}}}=0\)
for all \(\rno\in(0,1)\).
	\item\label{tilting-divergence}
	\(\RD{\rno}{\mW}{\qmn{\rno}}\),
	\(\RD{1}{\wma{\rno}{\qmn{\rno}}}{\mW}\), and 
	\(\RD{1}{\wma{\rno}{\qmn{\rno}}}{\qmn{\rno}}\) are continuous functions of \(\rno\)
	from \((0,1)\) to \(\reals{\geq0}\).
\end{enumerate}
\end{lemma}

Since \(\max_{\dinp\in\inpS} \RD{\rno}{\Wm(\dinp)}{\qmn{\rno,\Wm}}\) is finite 
by Lemma \ref{thm:minimax}
and \(\qmn{\rno,\Wm}\) changes continuously with \(\rno\)
by Lemma \ref{lem:centercontinuity}, one can invoke Lemma \ref{lem:tilting}
for \(\mW\!=\!\Wm(\dinp)\) and \(\qmn{\rno}\!=\!\qmn{\rno,\Wm}\)
for any \(\dinp\!\in\!\inpS\).
In the proof of Theorem \ref{proposition:SPBF}, 
this observation is used together with Lemma \ref{thm:tulcea}, given in 
the following, to construct a probability measure that is at the heart of the proof. 

For establishing Theorem \ref{proposition:SPBF},
we use two measure change arguments together with the Chebyshev's inequality. 
The bounds on the second moments, given in Lemma \ref{lem:variancebound}, 
are needed for applying the Chebyshev's inequality.  
\begin{lemma}[\!\!{\cite[Lemma 16.2-(a)]{augustin78}}]\label{lem:variancebound}
	If \(\RD{\rno}{\mW}{\mQ}\!<\!\infty\) for an \(\rno\!\in\!(0,1]\) and 
	\(\mW,\mQ\!\in\!\pdis{\outS}\), then
	\begin{align}
	\label{eq:lem:variancebound-w}
	\hspace{-.2cm}
	\sum\nolimits_{\dout} \wma{\rno}{\mQ}(\dout)
	\ln^{2}\tfrac{\wma{\rno}{\mQ}(\dout)}{\mW(\dout)}
	&\!\leq\!4 e^{-2}\!+\!\tfrac{(1-\rno)^2}{\rno^2}[4\!+\![\RD{\rno}{\mW}{\mQ}]^{2}],
	\\
	\label{eq:lem:variancebound-q}
	\hspace{-.2cm}
	\sum\nolimits_{\dout} \wma{\rno}{\mQ}(\dout)
	\ln^{2} \tfrac{\wma{\rno}{\mQ}(\dout)}{\mQ(\dout)}
	&\!\leq\!4 e^{-2}\!+\!\tfrac{4\rno^2}{(1-\rno)^2}
	\!+\![\RD{\rno}{\mW}{\mQ}]^{2}. 
	\end{align}
\end{lemma}
\begin{proof}[Proof of Lemma \ref{lem:variancebound}]
	Note that
	\begin{align}
	\label{eq:variancebound-1}
	\sum\nolimits_{\dout} \wma{\rno}{\mQ}(\dout)
	\ln^{2}\tfrac{\wma{\rno}{\mQ}(\dout)}{\mW(\dout)}
	\TCIND{[0,1]}{\tfrac{\wma{\rno}{\mQ}(\dout)}{\mW(\dout)}}
	&\!\leq\!4 e^{-2},
	\end{align}
	because  \(\sup_{\tau\in(0,1) }\tau\ln^{2}\tau=\left.\tau \ln^{2} \tau\right\vert_{\tau=e^{-2}}\leq 4 e^{-2}\).
	
	Furthermore, let \(\fX\!:\!\reals{\geq 0}\!\to\!\reals{+}\) be
	\begin{align}
\notag
	\fX(\tau)
	&=4 e^{-2}\tau \TCIND{[0,e^{2}]}{\tau}+\ln^{2}\tau \TCIND{(e^{2},\infty)}{\tau}.
	\end{align}
	Since \(\fX\) is a non-negative function  satisfying
	\(\ln^{2}\tau \leq \fX(\tau)\) for all  \(\tau\geq 1\)
	we have
	\begin{align}
	\notag
	\sum\nolimits_{\dout} 
	&\wma{\rno}{\mQ}(\dout)\ln^{2} \tfrac{\wma{\rno}{\mQ}(\dout)}{\mW(\dout)}
	\TCIND{(1,\infty)}{\tfrac{\wma{\rno}{\mQ}(\dout)}{\mW(\dout)}}
	\\
	\notag
	&\!=\!(\tfrac{1-\rno}{\rno})^{2}\sum\nolimits_{\dout} \wma{\rno}{\mQ}(\dout)
	\ln^{2} \left[\tfrac{\wma{\rno}{\mQ}(\dout)}{\mW(\dout)}\right]^{\frac{\rno}{1-\rno}}
	\TCIND{(1,\infty)}{\tfrac{\wma{\rno}{\mQ}(\dout)}{\mW(\dout)}}
	\\
	\label{eq:variancebound-3}
	&\!\leq\!(\tfrac{1-\rno}{\rno})^{2}\sum\nolimits_{\dout} \wma{\rno}{\mQ}(\dout)
	\fX\left(\left[\tfrac{\wma{\rno}{\mQ}(\dout)}{\mW(\dout)}\right]^{\frac{\rno}{1-\rno}}
	\right).
	\end{align}
	On the other hand the concavity of \(\fX\), 
	the Jensen's inequality, 
	the definition of tilted \pmf\!\!,
	and the monotonicity of \(\fX\) imply
	\begin{align}
	\notag
	\sum\nolimits_{\dout} \wma{\rno}{\mQ}(\dout)
	\fX\left(\left[\tfrac{\wma{\rno}{\mQ}(\dout)}{\mW(\dout)}\right]^{\frac{\rno}{1-\rno}}
	\right)
	&\!\leq\!\fX\left(\sum\nolimits_{\dout} \wma{\rno}{\mQ}(\dout)
	\left[\tfrac{\wma{\rno}{\mQ}(\dout)}{\mW(\dout)}\right]^{\frac{\rno}{1-\rno}}\right)
	\\
	\notag
	&\!\leq\! \fX\left(\sum\nolimits_{\dout} \mQ(\dout) e^{\RD{\rno}{\mW}{\mQ}}\right)
	\\
	\label{eq:variancebound-4}
	&\!\leq\! (2\vee \RD{\rno}{\mW}{\mQ})^2.
	\end{align}
	\eqref{eq:lem:variancebound-w} follows from 
	\eqref{eq:variancebound-1},
	\eqref{eq:variancebound-3}, \eqref{eq:variancebound-4}.
	One can prove \eqref{eq:lem:variancebound-q}, 
	following a similar analysis and invoking \eqref{eq:symmetry}.
\end{proof}

One can tilt the channel \(\Wm\!:\!\inpS\to\pdis{\outS}\) 
with a \(\mQ\) in \(\pdis{\outS}\), by tilting the individual 
\(\Wm(\dinp)\)'s; the resulting channel is called the tilted channel
and denoted by \(\Wma{\rno}{\mQ}\). 
If the R\'enyi center of the channel itself is used for tilting,
then we call the resulting channel the selftilted channel.
\begin{definition}\label{def:tiltedchannel}
	For any \(\Wm\!\in\!\pdis{\outS|\inpS}\) and \(\rno\in(0,1]\),
	\emph{the order-\(\rno\) selftilted channel \(\Wmn{\rno}\!:\!\inpS\!\to\!\pdis{\outS}\)}
	is 
	\begin{align}
	\notag
	\Wmn{\rno}(\dout|\dinp)
	&=[\Wm(\dout|\dinp)]^{\rno}[\qmn{\rno,\Wm}(\dout)]^{1-\rno}e^{(1-\rno)\RD{\rno}{\Wm(\dinp)}{\qmn{\rno,\Wm}}}
	\end{align}
	for all \(\dinp\!\in\!\inpS\) and \(\dout\!\in\!\outS\).
\end{definition}

\subsection{Construction of a Probability Measure with the Given Conditional Probabilities}\label{sec:tulcea}
In Definition \ref{def:Fproduct:discrete}, 
for describing the \pmf induced on the output set \(\outS_{1}^{\blx}\)  
by an element \(\vec{\dinp}_{1}^{\blx}\) of the input set  
\(\vec{\inpS}_{1}^{\blx}\),  
it was sufficient to specify the conditional \pmf given the past at each time instance.
This is true for arbitrary finite sample spaces, as well. 
When constructing probability measures in a similar
fashion for more general sample spaces, however,
there are additional technical conditions one needs to ensure.
If the conditional probability of events at each time are
Borel functions of the past, then the existence of a
unique probability measure is guaranteed, as demonstrated by 
the following lemma.

\begin{lemma}\label{thm:tulcea}
Let \((\smplS_{\tin},\smplG_{\tin})\) be an arbitrary measurable space 
for each  \(\tin\!\in\!\{1,\ldots,\blx\}\)
and \(\smplS\!=\!\smplS_{1}^{\blx}\), \(\smplG\!=\!\smplG_{1}^{\blx}\).
Suppose that a probability measure \(\prob^{(1)}\) is given on 
\((\smplS_{1},\smplG_{1})\) and that, 
for every \(\dsmpl_{1}^{\tin}\in \smplS_{1}^{\tin}\) and \(\tin\in\{1,\dots,\blx-1\}\),
probability measures 
\(\PCX{\cdot}{\dsmpl_{1}^{\tin}}\) are given on \((\smplS_{\tin+1},\smplG_{\tin+1})\).
Suppose that for every \(\set{B}\in\smplG_{\tin+1}\) the functions 
\(\PCX{\set{B}}{\dsmpl_{1}^{\tin}}\) are 
Borel functions of \(\dsmpl_{1}^{\tin}\) and let 
\begin{align}
\notag
\prob^{(\tin)}\!\left(\cset_{1}^{\tin}\right)
&\!=\!
\int_{\cset_{1}}\!\!\prob^{(1)}\!\left(\dif{\dsmpl_{1}}\right)
\int_{\cset_{2}}\!\!
\PCX{\dif{\dsmpl_{2}}}{\dsmpl_{1}}
\!\ldots\!
\int_{\cset_{\tin}}\!\!
\PCX{\dif{\dsmpl_{\tin}}}{\dsmpl_{1}^{\tin-1}}
\end{align}	
for all \(\cset_{\ind}\!\in\!\smplG_{\ind}\) and \(\tin\!\in\!\{2,\ldots,\blx\}\).
Then there is a unique probability measure \(\prob\) on 
\((\smplS,\smplG)\) such that
\begin{align}
\notag
\PX{\{\dsmpl:\dsmpl_{1}\!\in\!\cset_{1},\ldots,\dsmpl_{\tin}\!\in\!\cset_{\tin}\}}
&=\prob^{(\tin)}\!\left(\cset_{1}^{\tin}\right)
\end{align}
for every \(\tin\!\in\!\{1,\ldots,\blx\}\).
\end{lemma}

Lemma \ref{thm:tulcea} for \(n\!=\!2\) case is \cite[Thm. 2.6.2]{ash}.
For arbitrary but finite \(n\), Lemma \ref{thm:tulcea}  follows from
a recursive application of \cite[Thm. 2.6.2]{ash}. 
Lemma \ref{thm:tulcea} is also implied by 
Ionescu Tulcea's theorem \cite[Ch.II \S 9 Thm. 2]{shiryaev},
which establishes a more general result for the infinite 
horizon (i.e., \(n\)) case.

\begin{remark}
\(\PCX{\set{B}}{\dsmpl_{1}^{\tin}}\) is a Borel function
iff the inverse image of every Borel set is in \(\smplG_{1}^{\tin}\),
i.e., if
\(\{\dsmpl_{1}^{\tin}\!:\!\PCX{\set{B}}{\dsmpl_{1}^{\tin}}\!\in\!\set{C}\}\!\in\!\smplG_{1}^{\tin}\)
for every \(\set{C}\!\in\!\rborel{[0,1]}\).
If ---for example--- \((\smplS_{\tin},\smplG_{\tin})\!=\!(\reals{},\rborel{\reals{}})\)
for all \(\tin\), then \(\PCX{\set{B}}{\dsmpl_{1}^{\tin}}\)'s
are Borel functions
whenever \(\PCX{\set{B}}{\dsmpl_{1}^{\tin}}\) are continuous in \(\dsmpl_{1}^{\tin}\). 
\end{remark}
\begin{remark}
Lemma \ref{thm:tulcea} requires 
\(\PCX{\set{B}}{\dsmpl_{1}^{\tin}}\)'s to be 
Borel functions of \(\dsmpl_{1}^{\tin}\).
This is general enough for our purposes because 
we work with real valued random variables.
More generally, this condition is stated 
as the measurability of 
\(\PCX{\set{B}}{\dsmpl_{1}^{\tin}}\)
in \(\smplG_{1}^{\tin}\), 
which makes \(\PCX{\cdot}{\cdot}\)'s
transition probabilities (i.e., Markov kernels or stochastic kernels),
see \cite[\S10.7]{bogachev} for a more complete discussion.
The same measurability condition makes \(\PCX{\cdot}{\cdot}\)'s
conditional distributions in the sense of \cite[p. 343]{dudley},
as well.
\end{remark}

The proof of Theorem \ref{proposition:SPBF} presented
in the following section employs Lemma \ref{thm:tulcea} 
in order to assert the existence of a probability 
with certain conditional probabilities.
It is worth mentioning that we are not asserting that
one needs to consider infinite sample spaces in order to 
calculate the average error probability of a channel code
on a DSPC with feedback. 
The expressions in \eqref{eq:def:errorprobability} and
\eqref{eq:def:conditionalerrorprobability}
determine the value of the average error probability 
relying solely on a finite sample space model.
What we are saying is that Augustin's approach relies 
on a probability space with an infinite sample space 
in order to bound the minimum average
error probability of channel codes on a given DSPC 
with feedback. 

\subsection{Chebyshev's Inequality}\label{sec:chebyshev}
\begin{lemma}\label{lem:chebyshev}
	Let \(\amn{1},\ldots,\amn{\blx}\) be a sequence of real numbers and
	\((\inp_{1},\fltrn{1}),\ldots,(\inp_{\blx},\fltrn{\blx})\)
	be a stochastic sequence satisfying 
	\(\ECX{\inp_{\tin}}{\fltrn{\tin-1}}\!\leq\!\amn{\tin}\) and  
	\(\EX{(\inp_{\tin})^{2}}\!<\!\infty\)
	for all \(\tin\)  in \(\{1,\ldots,\blx\}\),
	and \(\sigma\) satisfy \(\sigma^{2}=\sum\nolimits_{\tin=1}^{\blx}\EX{(\inp_{\tin})^{2}}\). 
	Then 
	\begin{align}
	\label{eq:lem:chebyshev}
	\PX{\sum\nolimits_{\tin=1}^{\blx}\inp_{\tin}<\gamma+\sum\nolimits_{\tin=1}^{\blx}\amn{\tin}}
	&\geq 1- \tfrac{\sigma^{2}}{\gamma^2}
	\end{align}
for all \(\gamma\!\in\!\reals{+}\).
\end{lemma}
Lemma \ref{lem:chebyshev} is essentially a corollary of the Chebyshev's 
inequality,
a proof is presented in \hyperlink{appendix}{Appendix} for completeness.
A similar lemma was stated for a particular stochastic sequence 
and probability space in \cite[Lemma 41.4]{augustin78}.

\section{SPB for Codes on DSPCs with Feedback}\label{sec:SPB}
The main aim of this section is to prove a non-asymptotic SPB,
i.e., Lemma \ref{lem:SPBF}  given in the following.
We use this non-asymptotic SPB to prove the asymptotic one 
given in Theorem \ref{proposition:SPBF}
at the end of this section in  \S \ref{sec:proof:proposition:SPBF}.
Let us start with stating the aforementioned non-asymptotic SPB.  

\begin{lemma}\label{lem:SPBF}
	For any \(\Wm\!\in\!\pdis{\outS|\inpS}\) satisfying \(\RC{0}{\Wm}\!\neq\!\RC{1}{\Wm}\)
	and	\(\rate_{1},\!\rate_{2}\) satisfying \(\RC{0}{\Wm}\!<\!\rate_{0}\!<\!\rate_{1}\!<\!\RC{1}{\Wm}\),
	let	\(\rnf\!\in\!(0,1)\) satisfy \(\RC{\rnf}{\Wm}=\rate_{0}\),
	\(\rnt\in(\rnf,1)\) 
	satisfy\footnote{Such a \(\rnf\) and \(\rnt\) can always be found as a result of  Lemmas \ref{lem:capacityO} and \ref{lem:intermediate}.} 
	\(\tfrac{1-\rnt}{\rnt}\RC{\rnt}{\Wm}\!=\!\spe{\rate_{1},\!\Wm\!}\),
	positive parameter \(\epsilon\) satisfy \(\epsilon\leq \tfrac{\rnf\wedge(1-\rnt)}{2}\),
	and positive integers 	\(\blx\), \(\knd\) satisfy \(\knd\leq \blx\).
	Then any rate \(\rate\) channel code on 
	the length \(\blx\) DSPC with feedback
	\(\Wmn{\vec{[1,\blx]}}\) satisfying  \(\Wmn{\tin}=\Wm\)
	satisfies
	\begin{align}
	\label{eq:lem:SPBF}
	\Pem{av}
	&\!\geq\!e^{-\blx\left[\spe{\rate-\delta_{1},\Wm}+\delta_{2}\right]},
	\end{align}
	provided that 
	\begin{align}
	\label{eq:lem:SPBF-hypothesis}
	\hspace{-.2cm}
	\rate_{1}
	\!\geq\!\rate&\!\geq\! 
	\rate_{0}\!+\!\delta_{1},
	\end{align}
	where
	\begin{align}
	\label{eq:lem:SPBF-delta1}
	\delta_{1}
	&\DEF\tfrac{\ln 4}{\blx}+8	
	\tfrac{2+\RC{\sfrac{1}{2}}{\Wm}}{(1-\rnt)\sqrt{\knd}}
	+\tfrac{\knd}{\blx}\ln(\blx+\tfrac{1}{\epsilon}),
	\\
	\label{eq:lem:SPBF-delta2}
	\delta_{2}
	&\DEF\tfrac{\ln 4}{\blx}+8
	\tfrac{2+\RC{\sfrac{1}{2}}{\Wm}}{\rnf \sqrt{\knd}}
	+\tfrac{\knd \ln\blx}{\blx}+\tfrac{2\rate\epsilon}{\rnf^{2}}.
	\end{align}
\end{lemma}

The proof of Lemma \ref{lem:SPBF} relies on a pigeon hole argument 
and a measure change argument. In this respect, it is similar
to the standard proofs of the SPB.
Its principle novelty is in the choice/construction 
of the probability spaces and measures to apply these arguments.
We present this construction and the proof through self contained
pieces in \S\ref{sec:subblocks}-\S\ref{sec:pigeonhole}.
\begin{itemize}
\item In \S\ref{sec:subblocks}, we divide the block length into \(\knd\) subblocks of
approximately equal length.

\item In \S\ref{sec:construction}, we extend the natural finite sample space
that is used to describe the channel codes by introducing a positive valued
random variable at beginning of each subblock and construct probability 
measures \(\prob\), \(\prob_{\!\mV}\), \(\prob_{\!\mQ}\)
for the extended sample space
using a sequence of functions \(\gX_{1},\ldots,\gX_{\knd}\)
to be determined later.
The probability of the error event under \(\prob\) will be equal 
to \(\Pem{av}\) by construction.

\item In \S \ref{sec:tuning}, we describe a choice of the functions
\(\gX_{1},\ldots,\gX_{\knd}\) that bounds the order-one
R\'enyi divergence between the conditional \pmf\!\!'s of
the outputs of the subblocks, i.e. \(\out_{1+\tin_{\ind-1}}^{\tin_{\ind}}\!\)'s,
under \(\prob_{\!\mV}\) and \(\prob_{\!\mQ}\)
---as well as under \(\prob_{\!\mV}\) and \(\prob\)---
\(\prob_{\!\mV}\)-almost surely.

\item In \S\ref{sec:substantialprobability}, we use Chebyshev's inequality
to find an event \(\oev\) in the extended probability spaces
satisfying \(\PXS{\mV}{\!\oev\!}\!\geq\!0.5\)
for which both
\(\PX{\oev\cap \set{B}}\gtrapprox e^{-\blx \spe{\rate,\Wm}}\PXS{\mV}{\oev\cap \set{B}}\)
and
\(\PXS{\mQ}{\oev\cap \set{B}}\gtrapprox e^{-\blx \rate} \PXS{\mV}{\oev\cap \set{B}}\)
hold for any event \(\set{B}\) in the extended probability spaces.

\item In \S\ref{sec:pigeonhole}, we apply a measure change argument
together with a pigeon hole argument to prove Lemma \ref{lem:SPBF}.
\end{itemize}

In the following, we assume without loss of generality that 
the input and output sets are finite subsets of \(\reals{}\).
This will allow us to call 
the channel input and output at time \(\tin\) 
random variables and to denote them by \(\inp_{\tin}\)
and \(\out_{\tin}\), respectively.
Similarly, we assume that \(\mesS\) is a subset of \(\reals{}\) and denote
the random variables associated with the transmitted 
and decoded messages by \(\mes\) and \(\est\), respectively.  
We denote the realizations of the random variables 
such as \(\mes\), \(\sta_{\ind}\), \(\est\)
or vectors such as \(\inp_{\tau}^{\tin}\), \(\out_{\tau}^{\tin}\)
by the corresponding lower case letters 
such as \(\dmes\), \(\dsta_{\ind}\), \(\dest\)
or \(\dinp_{\tau}^{\tin}\), \(\dout_{\tau}^{\tin}\).
We denote the expected value of a random variable \(\qrv\)
under \(\prob_{\!\mV}\) by  \(\EXS{\mV}{\qrv}\).
As it is customary, we denote the expected value of 
a random variable \(\qrv\) 
conditioned on the random variable \(\sta\)
(i.e., conditioned on the minimum \(\sigma\)-algebra generated by \(\sta\))
by \(\ECX{\qrv}{\sta}\).
When we are working with \(\prob_{\!\mV}\)
instead of \(\prob\),
we use \(\ECXS{\mV}{\qrv}{\sta}\)
rather than \(\ECX{\qrv}{\sta}\).

\subsection{Division into \(\knd\) Subblocks}\label{sec:subblocks}
We divide the length \(\blx\) block into \(\knd\) subblocks of 
length either \(\lfloor\tfrac{\blx}{\knd}\rfloor\) or \(\lceil\tfrac{\blx}{\knd}\rceil\).
In particular, we set \(\tin_{0}\) to zero and define \(\ell_{\ind}\) and \(\tin_{\ind}\)
for \(\ind \in\{1,\ldots,\knd\}\) as follows
\begin{align}
\notag
\ell_{\ind}
&\DEF \lceil  \sfrac{\blx}{\knd} \rceil  \TCIND{(0,\blx-\lfloor \sfrac{\blx}{\knd} \rfloor \knd]}{\ind}+
\lfloor \sfrac{\blx}{\knd} \rfloor \TCIND{(\blx-\lfloor \sfrac{\blx}{\knd} \rfloor \knd,\knd]}{\ind},
\\
\notag
\tin_{\ind}
&\DEF \tin_{\ind-1}+\ell_{\ind}.
\end{align}
The last time instance of the \(\ind^{{th}}\) subblock is \(\tin_{\ind}\); 
for brevity, we denote the first time instance 
by \(\tau_{\ind}\), i.e.,
\begin{align}
\notag
\tau_{\ind}
&\DEF \tin_{\ind-1}+1.
\end{align}
Figure \ref{fig:time} demonstrates a typical partitioning of the length \(\blx\) block 
into $\knd$ subblocks.

\subsection{Construction of Auxiliary Probability Measures for a Given Sequence of Functions \(\gX_{1},\ldots,\gX_{\knd}\)}\label{sec:construction}
Let the sample space \(\smplS\) and \(\sigma\)-algebra of its subsets \(\smplA\) be
\begin{align}
\notag
\smplS
&\DEF\mesS\times\staS_{1}\times \outS_{\tau_{1}}^{\tin_{1}}\times
\cdots\times\staS_{\knd}\times\outS_{\tau_{\knd}}^{\tin_{\knd}},
\\
\notag
\smplA
&\DEF\sss{\mesS}\otimes\rborel{\staS_{1}}\otimes \sss{\outS_{\tau_{1}}^{\tin_{1}}}\otimes
\cdots\otimes\rborel{\staS_{\knd}}\otimes\sss{\outS_{\tau_{\knd}}^{\tin_{\knd}}},
\end{align}
where \(\staS_{\ind}\) is the open interval \((0,1)\) and 
\(\rborel{\staS_{\ind}}\)  is the associated Borel \(\sigma\)-algebra
for each \(\ind\) in \(\{1,\ldots,\knd\}\).

Let the \(\sigma\)-algebras \(\fltrn{0},\ldots,\fltrn{\knd}\) be
\begin{align}
\notag
\fltrn{0}
&\DEF \sss{\mesS},
\\
\notag
\fltrn{\ind}
&\DEF\fltrn{\ind-1}\otimes\rborel{\staS_{\ind}}\otimes\sss{\outS_{\tau_{\ind}}^{\tin_{\ind}}}
&
&\forall\ind\!\in\!\{1,\ldots,\knd\}.
\end{align}

In the following, we construct three probability measures on \((\smplS,\smplA)\)
---i.e., \(\prob\), \(\prob_{\!\mV}\), and  \(\prob_{\!\mQ}\)---
through their marginal distributions on \(\mesS\) and 
their  conditional distributions using Lemma \ref{thm:tulcea}.
The marginal distributions of \(\prob\), \(\prob_{\!\mV}\), and  \(\prob_{\!\mQ}\)
on the message set \(\mesS\) are all equal to the uniform distribution.
We specify the conditional distributions of \(\sta_{\ind}\!\)'s individually 
and the conditional distributions of \(\out_{\tin}\!\)'s
jointly through the conditional distributions of the vectors of
the form \(\out_{\tau_{\ind}}^{\tin_{\ind}}\).
In both cases, however, we demonstrate the conditional distributions to be Borel functions.
This allows us to invoke the existence of unique probability measures
\(\prob\), \(\prob_{\!\mV}\), and  \(\prob_{\!\mQ}\) on \((\smplS,\smplA)\) 
with the given conditional
distributions\footnote{Those readers who are not already familiar with the
	technical subtleties about the conditional probabilities 
	might benefit from taking the existence of 
	\(\prob\), \(\prob_{\!\mV}\), and  \(\prob_{\!\mQ}\) on \((\smplS,\smplA)\) 
	with the conditional distributions given in \eqref{eq:LSPBF_01}, 
	\eqref{eq:LSPBF_02w}, \eqref{eq:LSPBF_02q}, and \eqref{eq:LSPBF_02v}  
	granted, at least in their initial reading.}
 via  Lemma \ref{thm:tulcea}.

\begin{figure}[ht]
\begin{center}
	\begin{tikzpicture}
	\draw[decorate,decoration={brace,amplitude=11pt},xshift=0.0pt,yshift=0.0pt](0,1) -- (8.9,1) node [black,midway,yshift=0.5cm] 
	{\footnotesize $\blx$};
	
	\draw [thick] (4.85,0.5) -- (0,0.5) -- (0,1) -- (4.85,1);
	\draw [thick] (6.15,0.5) -- (8.9,0.5) -- (8.9,1) -- (6.15,1);
	
	\draw (0.4,1) -- (0.4,0.5);
	\node[right] at (-.053,.75) {\(\tau_{1}\)};
	\node at (1.0,.75) {\(\cdots\)};
	\draw (1.7,1) -- (1.7,0.5);
	\node[right] at (1.66,.75) {\(\tin_{1}\)};
	\draw (2.1,1) -- (2.1,0.5);
	\draw[decorate,decoration={brace,amplitude=8pt,mirror},xshift=0.0pt,yshift=0.0pt](0.0,0.5) -- (2.1,0.5) node [black,midway,yshift=-0.5cm] 
	{\footnotesize $\ell_{1}\!=\!\lceil\tfrac{\blx}{\knd}\rceil$};
	\draw (2.5,1) -- (2.5,0.5);
	\node[right] at (2.02,.75) {\(\tau_{2}\)};
	\node at (3.1,.75) {\(\cdots\)};
	\draw (3.8,1) -- (3.8,0.5);
	\node[right] at (3.76,.75) {\(\tin_{2}\)};
	\draw (4.2,1) -- (4.2,0.5);
	\draw[decorate,decoration={brace,amplitude=8pt,mirror},xshift=0.0pt,yshift=0.0pt](2.1,0.5) -- (4.2,0.5) node [black,midway,yshift=-0.5cm] 
	{\footnotesize $\ell_{2}$};
	
	\node at (5.45,.75) {\(\cdots\cdots\)};
	
	\draw (6.9,1) -- (6.9,0.5);
	\node[right] at (6.82,.75) {\(\tau_{\knd}\)};
	\draw (7.3,1) -- (7.3,0.5);
	\node at (8.0,.75) {\(\cdots\)};
	\draw (8.5,1) -- (8.5,0.5);
	\node[right] at (8.43,.75) {\(\tin_{\knd}\)};
	\draw[decorate,decoration={brace,amplitude=8pt,mirror},xshift=0.0pt,yshift=0.0pt](6.9,0.5) -- (8.9,0.5) node [black,midway,yshift=-0.5cm] 
	{\footnotesize $\ell_{\knd}\!=\!\lfloor\tfrac{\blx}{\knd}\rfloor$};
	\end{tikzpicture}
	\vspace{-.6cm}
\caption{A typical partitioning of the length \(\blx\) block into \(\knd\) subblocks. 		
The length of the first subblock is always $\lceil\tfrac{\blx}{\knd}\rceil$
and the length of the last subblock is always $\lfloor\tfrac{\blx}{\knd}\rfloor$.}
\label{fig:time}
\vspace{-.3cm}
\end{center}
\end{figure}
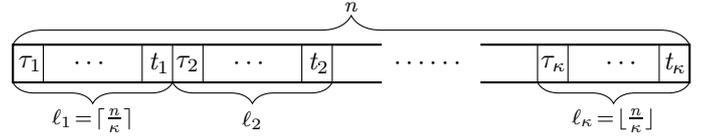 

Let us first describe the conditional distributions of
\(\sta\)'s.
Let \(\gX_{1}\) be a function from \(\mesS\) to \((0,1)\)
to be determined later.
Similarly, for each \(\ind\)  in \(\{2,\ldots,\blx\}\), 
let
\(\gX_{\ind}\!:\!\mesS\times\outS_{1}^{\tin_{\ind-1}}\to(0,1)\)
be a  function that is to be determined later.
The conditional distribution of \(\sta_{\ind}\) is 
the same for \(\prob\), \(\prob_{\!\mV}\), and  \(\prob_{\!\mQ}\)
and it is determined by the function \(\gX_{\ind}\) as follows:
\vspace{-.1cm}
\begin{align}
\label{eq:LSPBF_01}
\PCX{\cset}{\dmes,\dsta_{1}^{\ind-1},\dout_{1}^{\tin_{\ind-1}}}	
&\!=\!\tfrac{1}{\epsilon}
\int_{(1-\epsilon)\rno}^{\rno+\epsilon(1-\rno)} \TCIND{\cset}{\dsta} \dif{\dsta}
\end{align}
for all \(\cset\!\in\!\rborel{\staS_{\ind}}\),
where \(\rno\!=\!\gX_{\ind}(\dmes,\dout_{1}^{\tin_{\ind-1}})\).
Since \(\mesS\times\outS_{1}^{\tin_{\ind-1}}\) is a finite set,
all of the elements of its power set are Borel sets
and \(\PCX{\cset}{\dmes,\dsta_{1}^{\ind-1},\dout_{1}^{\tin_{\ind-1}}}\)
is a Borel function for any \(\cset\!\in\!\rborel{\staS_{\ind}}\).

Let us proceed with the description of the conditional probability 
distributions of  \(\out\)'s.
For \(\prob\) we have
\begin{align}
\label{eq:LSPBF_02w}
\PCX{\dout_{\tau_{\ind}}^{\tin_{\ind}}}{\dmes,\dsta_{1}^{\ind},\dout_{1}^{\tin_{\ind-1}}}	
&\!=\!
\prod\nolimits_{\tin=\tau_{\ind}}^{\tin_{\ind}}\!\Wm(\dout_{\tin}|\dinp_{\tin})
\end{align}
for all \(\dout_{\tau_{\ind}}^{\tin_{\ind}}\!\in\!\outS_{\tau_{\ind}}^{\tin_{\ind}}\)
where \(\dinp_{\tin}\) is the channel input at time 
\(\tin\), which is nothing but \(\vec{\dinp\!}_{\!\tin}(\dout_{1}^{\tin-1})\) 
for \(\vec{\dinp\!}_{1}^{\blx}\) 
satisfying \(\enc(\dmes)\!=\!\vec{\dinp\!}_{1}^{\blx}\).
Note that \(\PCX{\cset}{\dmes,\dsta_{1}^{\ind},\dout_{1}^{\tin_{\ind-1}}}\) does not
depend on \(\dsta_{1}^{\ind}\).
Thus \(\PCX{\cset}{\dmes,\dsta_{1}^{\ind},\dout_{1}^{\tin_{\ind-1}}}\)
is a Borel function for all 
\(\cset\!\subset\!\outS_{\tau_{\ind}}^{\tin_{\ind}}\)
as a consequence of the finiteness of \(\mesS\times\outS_{1}^{\tin_{\ind-1}}\). 

For \(\prob_{\!\mQ}\) we have
\vspace{-.1cm}
\begin{align}
\label{eq:LSPBF_02q}
\PCXS{\mQ}{\dout_{\tau_{\ind}}^{\tin_{\ind}}}{\dmes,\dsta_{1}^{\ind},\dout_{1}^{\tin_{\ind-1}}}	
&\!=\!
\prod\nolimits_{\tin=\tau_{\ind}}^{\tin_{\ind}}\!
\qmn{\dsta_{\ind},\Wm}(\dout_{\tin})
\end{align}
for all \(\dout_{\tau_{\ind}}^{\tin_{\ind}}\!\in\!\outS_{\tau_{\ind}}^{\tin_{\ind}}\).
Since R\'enyi center is continuous in its order by Lemma \ref{lem:centercontinuity}, 
\(\PCXS{\mQ}{\cset}{\dmes,\dsta_{1}^{\ind},\dout_{1}^{\tin_{\ind-1}}}\)
is a continuous and hence a Borel function of \(\dsta_{\ind}\)
for all  \(\cset\!\subset\!\outS_{\tau_{\ind}}^{\tin_{\ind}}\).

For \(\prob_{\!\mV}\) we have
\vspace{-.1cm}
\begin{align}
\label{eq:LSPBF_02v}
\PCXS{\mV}{\dout_{\tau_{\ind}}^{\tin_{\ind}}}{\dmes,\dsta_{1}^{\ind},\dout_{1}^{\tin_{\ind-1}}}	
&\!=\!
\prod\nolimits_{\tin=\tau_{\ind}}^{\tin_{\ind}}\!\Wmn{\dsta_{\ind}}(\dout_{\tin}|\enc_{\tin}(\dmes,\dout_{1}^{\tin-1}))
\end{align}
for all \(\dout_{\tau_{\ind}}^{\tin_{\ind}}\!\in\!\outS_{\tau_{\ind}}^{\tin_{\ind}}\)
where \(\Wmn{\dsta_{\ind}}\) is the order-\(\dsta_{\ind}\) 
selftilted channel described in Definition \ref{def:tiltedchannel}
and \(\dinp_{\tin}\) is the channel input at time 
\(\tin\). 
Since \(\Wmn{\rno}(\cdot|\dinp)\) is continuous in \(\rno\)
for any \(\dinp\)  
by Lemmas \ref{lem:centercontinuity} and \ref{lem:tilting},
\(\PCXS{\mV}{\cset}{\dmes,\dsta_{1}^{\ind},\dout_{1}^{\tin_{\ind-1}}}\) 
is a continuous function of \(\dsta_{\ind}\) for any \(\dout_{1}^{\tin_{\ind-1}}\),
which does not depend on \(\dsta_{1}^{\ind-1}\). 
Since \(\outS_{1}^{\tin_{\ind-1}}\) is a finite set,
this will ensure 
\(\PCXS{\mV}{\cset}{\dmes,\dsta_{1}^{\ind},\dout_{1}^{\tin_{\ind-1}}}\) 
to be a Borel function for any \(\cset\!\subset\!\outS_{\tau_{\ind}}^{\tin_{\ind}}\).
\vspace{-.3cm}

\subsection{A Choice of \(\gX_{1},\ldots,\gX_{\knd}\)}\label{sec:tuning}
The preceding construction works for any choice of the functions 
\(\gX_{1},\ldots,\gX_{\knd}\).
However, only some of the choices are appropriate for our purposes.
In the following, we choose 
\(\gX_{1},\ldots,\gX_{\knd}\) by determining the value of 
\(\gX_{\ind}(\dmes,\dout_{1}^{\tin_{\ind-1}})\) 
for each \(\ind\), \(\dmes\), and \(\dout_{1}^{\tin_{\ind-1}}\)
individually 
and commit to the resulting \(\gX_{1},\ldots,\gX_{\knd}\)'s
for the rest of the paper.
In order to find the aforementioned appropriate choice we analyze the 
value of certain conditional expectation
---i.e., \(\ECXS{\mV}{\hrv_{\ind}}{\dmes,\dout_{1}^{\tin_{\ind-1}}}\)---
as a function of the value of  \(\gX_{\ind}\) at 
\((\dmes,\dout_{1}^{\tin_{\ind-1}})\)
---i.e., as a function of 
\(\gX_{\ind}(\dmes,\dout_{1}^{\tin_{\ind-1}})\)--- 
at each \((\dmes,\dout_{1}^{\tin_{\ind-1}})\) individually.

Note that \(\RD{1}{\Wmn{\sta_{\ind}}(\inp_{\tin})}{\qmn{\sta_{\ind},\Wm}}\)
is a random variable that is measurable in the \(\sigma\)-algebra generated by
\(\inp_{\tin}\) and \(\sta_{\ind}\)
because \(\RD{1}{\Wmn{\dsta}(\dinp)}{\qmn{\dsta,\Wm}}\) is 
continuous in \(\dsta\) by Lemmas \ref{lem:centercontinuity}  and  \ref{lem:tilting}.
For any \(\ind\!\in\!\{1,\ldots,\knd\}\),
let the random variable \(\hrv_{\ind}\) be
\begin{align}
\label{eq:LSPBF_03}
\hrv_{\ind}
&\!\DEF\!
\sum\nolimits_{\tin=\tau_{\ind}}^{\tin_{\ind}}
\ECXS{\mV}{\RD{1}{\Wmn{\sta_{\ind}}(\inp_{\tin})}{\qmn{\sta_{\ind},\Wm}}}{\fltrn{\ind-1},\!\sta_{\ind}}.
\end{align}
Note that
 \(\hrv_{\ind}\) is a non-negative 
random variable 
by \eqref{eq:pinsker}.
Furthermore
\(\RD{1}{\!\Wmn{\sta_{\ind}}(\inp_{\tin})}{\qmn{\sta_{\ind},\Wm}\!}\!\leq\!\RD{\sta_{\ind}}{\!\Wm(\inp_{\tin})}{\qmn{\sta_{\ind},\Wm}\!}\)
by \eqref{eq:pinsker} and \eqref{eq:tiltedKLD} and  
\(\RD{\sta_{\ind}}{\!\Wm(\inp_{\tin})}{\qmn{\sta_{\ind},\Wm}\!}\!\leq\!\RC{\sta_{\ind}}{\Wm}\)
by Lemma \ref{thm:minimax}.
Thus for any \(\ind\!\in\!\{1,\ldots,\knd\}\),
the random variables \(\hrv_{\ind}\) and \(\RC{\sta_{\ind}}{\Wm}\) 
satisfy
\begin{align}
\label{eq:LSPBF_04}
0\leq \hrv_{\ind}
&\leq \ell_{\ind} \RC{\sta_{\ind}}{\Wm}
&
&
\end{align}
for all realizations of \(\fltrn{\ind-1}\) and \(\sta_{\ind}\).
Then for all realizations of \(\fltrn{\ind-1}\),
the conditional expectation \(\ECXS{\mV}{\hrv_{\ind}}{\fltrn{\ind-1}}\)
is a continuous function of the value of \(\gX_{\ind}\) 
at \((\dmes,\dout_{1}^{\tin_{\ind-1}})\)
---i.e., \(\gX_{\ind}(\dmes,\dout_{1}^{\tin_{\ind-1}})\)---
as a result of \eqref{eq:LSPBF_01} defining the conditional 
distribution of \(\sta_{\ind}\) for \(\prob\), \(\prob_{\!\mV}\), and  \(\prob_{\!\mQ}\),
because  \(\RC{\rno}{\Wm}\)
is nondecreasing in \(\rno\) and finite on \((0,1)\) by Lemma \ref{lem:capacityO}.
Thus we can tune the value of 
\(\ECXS{\mV}{\hrv_{\ind}}{\fltrn{\ind-1}}\) by changing the value of the 
function \(\gX_{\ind}\) for different realizations of \(\mes\)
and \(\out_{1}^{\tin_{\ind-1}}\).

On the other hand as a result of the construction, we have
\begin{align}
\label{eq:LSPBF_05}
\ECXS{\mV}{\hrv_{\ind}}{\fltrn{\ind-1}}
&=\ECXS{\mV}{\hrv_{\ind}}{\mes,\out_{1}^{\tin_{\ind-1}}}.
\end{align}

We use the following rule to choose the value \(\gX_{\ind}\) 
at each \((\dmes,\dout_{1}^{\tin_{\ind-1}})\)
depending on the rate of the code \(\rate\)
and the positive constant \(\delta_{1}\) defined in \eqref{eq:lem:SPBF-delta1}.
\begin{itemize}
	\item If \(\ECXS{\mV}{\hrv_{\ind}}{\dmes,\dout_{1}^{\tin_{\ind-1}}}
	\leq \ell_{\ind}(\rate-\delta_{1})\) 
	for \(\gX_{\ind}(\dmes,\dout_{1}^{\tin_{\ind-1}})\!=\!\tfrac{\rnt}{1-\epsilon}\),
	then \(\gX_{\ind}(\dmes,\dout_{1}^{\tin_{\ind-1}})\!=\!\tfrac{\rnt}{1-\epsilon}\).

	\item If \(\ECXS{\mV}{\hrv_{\ind}}{\dmes,\dout_{1}^{\tin_{\ind-1}}}>\ell_{\ind}(\rate-\delta_{1})\) 
for \(\gX_{\ind}(\dmes,\dout_{1}^{\tin_{\ind-1}})\!=\!\tfrac{\rnt}{1-\epsilon}\),
then \(\gX_{\ind}(\dmes,\dout_{1}^{\tin_{\ind-1}})\!=\!\rno\) for an \(\rno\) in 
\([\tfrac{\rnf-\epsilon}{1-\epsilon},\tfrac{\rnt}{1-\epsilon})\) satisfying 
\(\ECXS{\mV}{\hrv_{\ind}}{\dmes,\dout_{1}^{\tin_{\ind-1}}}\!=\!\ell_{\ind}(\rate-\delta_{1})\).
The existence of such an \(\rno\) follows from 
the continuity of \(\ECXS{\mV}{\hrv_{\ind}}{\dmes,\dout_{1}^{\tin_{\ind-1}}}\)
in the value of \(\gX_{\ind}(\dmes,\dout_{1}^{\tin_{\ind-1}})\), 
the intermediate value theorem \cite[4.23]{rudin},  
and the inequality  
\(\ECXS{\mV}{\hrv_{\ind}}{\dmes,\dout_{1}^{\tin_{\ind-1}}}\leq\ell_{\ind}(\rate-\delta_{1})\)
for \(\gX_{\ind}(\dmes,\dout_{1}^{\tin_{\ind-1}})\!=\!\tfrac{\rnf-\epsilon}{1-\epsilon}\).
In order to see why the inequality at \(\tfrac{\rnf-\epsilon}{1-\epsilon}\)  holds,
first note that \eqref{eq:LSPBF_01} and \eqref{eq:LSPBF_04} imply
\vspace{-.1cm}
\begin{align}
\notag
\ECXS{\mV}{\hrv_{\ind}}{\dmes,\dout_{1}^{\tin_{\ind-1}}}
\!\leq\!\tfrac{\ell_{\ind}}{\epsilon}
\int_{\rnf-\epsilon}^{\rnf} \RC{\dsta}{\Wm} \dif{\dsta}.
\end{align}
Then the inequality follows from \eqref{eq:lem:SPBF-hypothesis}, 
\(\rate_{0}\!=\!\RC{\rnf}{\Wm}\),
and the monotonicity of the R\'enyi capacity in its order.
\end{itemize}
The choice of \(\gX_{\ind}\)'s described 
above ensures not only
\begin{align}
\label{eq:LSPBF_06}
0\leq \ECXS{\mV}{\hrv_{\ind}}{\fltrn{\ind-1}}
&\leq \ell_{\ind}(\rate-\delta_{1})
\end{align}
for all \(\ind\in\{1,\ldots,\knd\}\), but also
\begin{align}
\label{eq:LSPBF_07}
\tfrac{\rnf-\epsilon}{1-\epsilon}
\leq \grv_{\ind}
&\leq \tfrac{\rnt}{1-\epsilon}
\\
\label{eq:LSPBF_08}
\rnf-\epsilon
\leq \sta_{\ind}
&\leq \rnt+\epsilon
\end{align}
for all \(\ind\in\{1,\ldots,\knd\}\),
where \(\grv_{\ind}\) is the random
variable defined as \(\grv_{\ind}\DEF\gX_{\ind}(\mes,\out_{1}^{\tin_{\ind-1}})\).

\subsection{Application of Chebyshev's Inequality to Find an Event with Substantial Probability under the Auxiliary Measure}\label{sec:substantialprobability}
The preceding choice of the functions \(\gX_{1},\ldots,\gX_{\knd}\), 
bounds the expected value of random 
variables that are used in the measure change argument. 
In order to apply the measure change argument, we first prove
---using Lemma \ref{lem:chebyshev}--- that
these random variables take values that are close to their means 
with substantial probability under \(\prob_{\!\mV}\).

Let \(\fltrn{\ind}\)-measurable random variable
\(\qrv_{\ind}\) be
\begin{align}
\label{eq:LSPBF_09}
\qrv_{\ind}
&\!\DEF\!\ln \tfrac{\PCXS{\mV}{\out_{\tau_{\ind}}^{\tin_{\ind}}}{\mes,\sta_{1}^{\ind},\out_{1}^{\tin_{\ind-1}}}}
{\PCXS{\mQ}{\out_{\tau_{\ind}}^{\tin_{\ind}}}{\mes,\sta_{1}^{\ind},\out_{1}^{\tin_{\ind-1}}}}
&
&~
&
&~
\forall \ind\!\in\!\{1,\ldots,\knd\}.
\end{align} 
Note that 
\eqref{eq:LSPBF_02q},
\eqref{eq:LSPBF_02v},
\eqref{eq:LSPBF_03},
and the definition of order-one R\'enyi divergence imply
\begin{align}
\notag
\ECXS{\mV}{\qrv_{\ind}}{\fltrn{\ind-1},\!\sta_{\ind}}
&\!=\!\hrv_{\ind}. 
\end{align}
Then \eqref{eq:LSPBF_06} implies
\begin{align}
\label{eq:LSPBF_10}
0\leq \ECXS{\mV}{\qrv_{\ind}}{\fltrn{\ind-1}}
&\leq \ell_{\ind}(\rate-\delta_{1})
&
&~
&
&~
\forall \ind\!\in\!\{1,\ldots,\knd\}.
\end{align}

Let us proceed with bounding the second moments  
of \(\qrv_{\ind}\!\)'s from above. The Cauchy-Schwarz inequality implies
\begin{align}
\notag
\EXS{\mV}{(\qrv_{\ind})^{2}}
&\!=\!
\sum\nolimits_{\tin=\tau_{\ind}}^{\tin_{\ind}}
\sum\nolimits_{\bar{\jnd}=\tau_{\ind}}^{\tin_{\ind}}
\EXS{\mV}{
	\drv_{\tin}
	\drv_{\jnd}}
\\
\notag
&\!\leq\!
\sum\nolimits_{\tin=\tau_{\ind}}^{\tin_{\ind}}
\sum\nolimits_{\bar{\jnd}=\tau_{\ind}}^{\tin_{\ind}}
\sqrt{\EXS{\mV}{(\drv_{\tin})^{2}}
	\EXS{\mV}{(\drv_{\jnd})^{2}}
}
\end{align}
where 
\(\!\drv_{\tin}
\!\DEF\!\ln\!\tfrac{\PCXS{\mV}{\out_{\tin}}{\mes,\sta_{1}^{\ind},\out_{1}^{\tin-1}}}
{\PCXS{\mQ}{\out_{\tin}}{\mes,\sta_{1}^{\ind},\out_{1}^{\tin-1}}}\)
for all \(\tin\in\{\tau_{\ind},\ldots,\tin_{\ind}\}\).

On the other hand using 
the definition of the order-one R\'enyi divergence
and  
\eqref{eq:lem:variancebound-q} of Lemma \ref{lem:variancebound}
we get
\begin{align}
\notag
\EXS{\mV}{(\drv_{\tin})^{2}}
&=\EXS{\mV}{
	\ECXS{\mV}{(\drv_{\tin})^{2}}{\fltrn{\ind-1},\sta_{\ind},\out_{\tau_{\ind}}^{\tin-1}}}
\\
\notag
&\leq 
\EXS{\mV}{\tfrac{4}{(1-\sta_{\ind})^{2}}
	\!+\!(\RD{\sta_{\ind}}{\Wm(\inp_{\tin})}{\qmn{\sta_{\ind},\Wm}})^{2}}
\end{align}
First invoking 
\eqref{eq:thm:minimaxradiuscenter}  and
\eqref{eq:orderoneovertwo} 
to bound \(\RD{\sta_{\ind}}{\Wm(\inp_{\tin})}{\qmn{\sta_{\ind},\Wm}}\),
and then using the identity
\(1-\sta_{\ind}\geq \tfrac{1-\rnt}{2}\),
which follows from \eqref{eq:LSPBF_08}
and the hypothesis
\(\epsilon\leq \tfrac{\rnf\wedge(1-\rnt)}{2}\),
we get
\begin{align}
\notag
\EXS{\mV}{(\drv_{\tin})^{2}}
&\!\leq\!
\EXS{\mV}{\tfrac{4+(\RC{\sfrac{1}{2}}{\Wm})^{2}}{(1-\sta_{\ind})^{2}}}
\\
\notag
&\!\leq\!
4\tfrac{4+(\RC{\sfrac{1}{2}}{\Wm})^{2}}{(1-\rnt)^{2}}.
\end{align}
Thus using \(\ell_{\ind} \leq 2\tfrac{\blx}{\knd}\) we get
\begin{align}
\notag
\EXS{\mV}{(\qrv_{\ind})^{2}}
&\leq\ell_{\ind}^{2} 4\tfrac{4+(\RC{\sfrac{1}{2}}{\Wm})^{2}}{(1-\rnt)^{2}}
\\
\label{eq:LSPBF_11}
&\leq 16\tfrac{4+(\RC{\sfrac{1}{2}}{\Wm})^{2}}{(1-\rnt)^{2}} \tfrac{\blx^{2}}{\knd^{2}}.
\end{align}
Applying Lemma \ref{lem:chebyshev},
for \(\amn{\ind}=\ell_{\ind}(\rate-\delta_{1})\)
to the stochastic 
sequence\footnote{Note that \(\smplA_{\ind}\)'s are not
	defined as \(\sigma\)-algebras on \(\smplS\) and hence they are not 
	sub-\(\sigma\)-algebras of \(\smplA\).
	Nevertheless, for each \(\smplA_{\ind}\)
	there is a corresponding \(\widetilde{\fltrn{\ind}}\subset \smplA\)
	that uniquely determines \(\fltrn{\ind}\) and that is uniquely determined by \(\fltrn{\ind}\).
	When applying Lemma \ref{lem:chebyshev} we are in fact considering 
	\((\qrv_{1},\widetilde{\smplA}_{1}),\ldots,(\qrv_{\knd},\widetilde{\smplA}_{\knd})\)
	rather than
	\((\qrv_{1},\smplA_{1}),\ldots,(\qrv_{\knd},\smplA_{\knd})\).}  
	\((\qrv_{1},\smplA_{1}),\ldots,(\qrv_{\knd},\smplA_{\knd})\)
via \eqref{eq:LSPBF_10} we get
\begin{align}
\notag
\PXS{\mV}{\qrv\!\leq\!\blx (\rate-\delta_{1})+\gamma}
&\geq 1-\tfrac{\sum\nolimits_{\ind=1}^{\knd}\EXS{\mV}{(\qrv_{\ind})^{2}}}{\gamma^{2}},
\end{align}
where \(\qrv\) is defined as
\begin{align}
\label{eq:LSPBF_Qdef}
\qrv
&\DEF\sum\nolimits_{\ind=1}^{\knd} \qrv_{\ind}.
\end{align}
Setting 
\(\gamma=8\tfrac{(2+\RC{\sfrac{1}{2}}{\Wm})\blx}{(1-\rnt)\sqrt{\knd}}\)
and invoking
\eqref{eq:lem:SPBF-delta1} and \eqref{eq:LSPBF_11} we get
\begin{align}
\label{eq:LSPBF-Q-MC-prob}
\PXS{\mV}{\oev_{\mQ}}
&\geq \tfrac{3}{4},
\end{align}
where \(\oev_{\mQ}\) is defined as
\begin{align}
\label{eq:LSPBF-Q-MC-event}
\oev_{\mQ}
&\!\DEF\!\left\{\dsmpl\!\in\!\smplS\!:\!
\qrv(\dsmpl)\!\leq\!\blx \rate\!-\!\ln 4\!-\!\knd \ln\left(\blx+\tfrac{1}{\epsilon}\right)
\right\}.
\end{align}
Recall that for all \(\ind\in\{1,\ldots,\knd\}\) 
the conditional distributions of \(\prob_{\!\mV}\) and  \(\prob_{\!\mQ}\)
for \(\sta_{\ind}\)'s given \(\smplA_{\ind-1}\) are identical 
because of \eqref{eq:LSPBF_01}.
Thus \(\qrv(\dsmpl)=\ln \der{\prob_{\!\mV}}{\prob_{\!\mQ}}(\dsmpl)\) and consequently 
\begin{align}
\label{eq:LSPBF-Q-MC-bound}
\PXS{\mQ}{\set{B}\cap \left\{\qrv\leq \lambda\right\}}
&\geq e^{-\lambda}\PXS{\mV}{\set{B}\cap \left\{\qrv\leq \lambda\right\}}
\end{align}
for any \(\set{B}\in\smplA\) and \(\lambda\in\reals{}\).

We need identities analogous to 
\eqref{eq:LSPBF-Q-MC-prob} and  \eqref{eq:LSPBF-Q-MC-bound} 
for \(\prob\) and \(\prob_{\!\mV}\), 
as well. The random variables \(\vrv_{1},\ldots,\vrv_{\knd}\)
are used to obtain those identities.
For any \(\ind\!\in\!\{1,\ldots,\knd\}\),
let \(\fltrn{\ind}\)-measurable random variable
\(\vrv_{\ind}\) be
\begin{align}
\label{eq:LSPBF_13}
\vrv_{\ind}
&\!\DEF\!\ln \tfrac{\PCXS{\mV}{\out_{\tau_{\ind}}^{\tin_{\ind}}}{\mes,\sta_{1}^{\ind},\out_{1}^{\tin_{\ind-1}}}}
{\PCX{\out_{\tau_{\ind}}^{\tin_{\ind}}}{\mes,\sta_{1}^{\ind},\out_{1}^{\tin_{\ind-1}}}}.
\end{align}
Then as a result of 
\eqref{eq:LSPBF_02w},
\eqref{eq:LSPBF_02v},
and the definition of order-one R\'enyi divergence
\begin{align}
\notag
\ECXS{\mV}{\vrv_{\ind}}{\fltrn{\ind-1},\!\sta_{\ind}}
&\!=\!
\sum\nolimits_{\tin=\tau_{\ind}}^{\tin_{\ind}}\!
\!\!\ECXS{\mV}{\RD{1}{\!\Wmn{\sta_{\ind}}(\inp_{\tin})}{\!\Wm(\inp_{\tin})}}{\fltrn{\ind-1},\!\sta_{\ind}}.
\end{align}
On the other hand 
as a result of \eqref{eq:tiltedKLD} and Lemma \ref{thm:minimax},
we have
\begin{align}
\notag
\RD{1}{\!\Wmn{\sta_{\ind}}(\inp_{\tin})}{\!\Wm(\inp_{\tin})}
\!\leq\!\tfrac{1-\sta_{\ind}}{\sta_{\ind}}\left(\RC{\sta_{\ind}}{\!\Wm}\!-\!\RD{1}{\!\Wmn{\sta_{\ind}}(\inp_{\tin})}{\qmn{\sta_{\ind},\Wm}}\right)
\end{align}
for all  \(\tin\in\{\tau_{\ind},\ldots,\tin_{\ind}\}\).

Then the non-negativity of the R\'enyi divergence
and the definition of \(\hrv_{\ind}\) 
given in \eqref{eq:LSPBF_03}
imply
\begin{align}
\label{eq:LSPBF_14}
0\leq \ECXS{\mV}{\vrv_{\ind}}{\fltrn{\ind-1}}
&\leq 
\ECXS{\mV}{
\tfrac{1-\sta_{\ind}}{\sta_{\ind}}\left(\ell_{\ind}\RC{\sta_{\ind}}{\Wm}-\hrv_{\ind}\right)}{\fltrn{\ind-1}}.
\end{align}
We bound the expression on the right hand side of \eqref{eq:LSPBF_14}
through a case by case analysis based on the value of \(\grv_{\ind}\).
\begin{itemize}
\item If \(\grv_{\ind}\!=\!\tfrac{\rnt}{1-\epsilon}\),
then \(\sta_{\ind}\geq \rnt\) by construction.
On the other hand
\(\tfrac{1-\rnt}{\rnt}\RC{\rnt}{\Wm}\!=\!\spe{\rate_{1},\!\Wm\!}\)
by the hypothesis and 
\(\tfrac{1-\rno}{\rno}\RC{\rno}{\Wm}\) is nonincreasing in \(\rno\)
by Lemma \ref{lem:capacityO}.
Thus
\(\ECXS{\mV}{\vrv_{\ind}}{\fltrn{\ind-1}}\leq \spe{\rate_{1},\!\Wm\!}\)
as a result of the non-negativity of \(\hrv_{\ind}\) established by \eqref{eq:LSPBF_04}.
Since \(\spe{\rate,\!\Wm\!}\) is nonincreasing  in \(\rate\) by definition
we get
\begin{align}
\label{eq:LSPBF_15}
\ECXS{\mV}{
	\tfrac{1-\sta_{\ind}}{\sta_{\ind}}\left(\ell_{\ind}\RC{\sta_{\ind}}{\Wm}-\hrv_{\ind}\right)}{\fltrn{\ind-1}}
&\leq \ell_{\ind}\!\spe{\rate-\delta_{1},\Wm\!}.
\end{align}
\item If \(\grv_{\ind}\!\neq\!\tfrac{\rnt}{1-\epsilon}\),
then \(\ECXS{\mV}{\hrv_{\ind}}{\fltrn{\ind-1}}\!=\!\ell_{\ind}(\rate-\delta_{1})\)
by construction. Thus 
\(\hrv_{\ind}\geq0\) ---established in \eqref{eq:LSPBF_04}---
and \eqref{eq:LSPBF_01} imply 
\begin{align}
\notag
&\ECXS{\mV}{
\tfrac{1-\sta_{\ind}}{\sta_{\ind}}\left(\ell_{\ind}(\rate-\delta_{1})-\hrv_{\ind}\right)}{\fltrn{\ind-1}}
\\
\notag
&\qquad\!\leq\!\tfrac{1-(1-\epsilon)\grv_{\ind}}{(1-\epsilon)\grv_{\ind}}\ell_{\ind}(\rate-\delta_{1})
-\!\tfrac{(1-\epsilon)(1-\grv_{\ind})}{\grv_{\ind}+\epsilon(1-\grv_{\ind})}
\ECXS{\mV}{\hrv_{\ind}}{\fltrn{\ind-1}}
\\
\notag
&\qquad\!=\tfrac{\ell_{\ind} (\rate-\delta_{1})\epsilon }{(\grv_{\ind}-\epsilon\grv_{\ind})(\grv_{\ind}+\epsilon(1-\grv_{\ind}))}.
\end{align}
On the other hand \(\grv_{\ind}\!\geq\!\tfrac{\rnf-\epsilon}{1-\epsilon}\)
by \eqref{eq:LSPBF_07}, \(\epsilon\!\leq\!\tfrac{\rnf}{2}\) 
by hypothesis
and
\(\tfrac{1-\sta_{\ind}}{\sta_{\ind}}\!\left(\RC{\sta_{\ind}}{\Wm}\!-\!(\rate-\delta_{1})\right)\leq\spe{\rate-\delta_{1},\Wm}\) 
by the definition of \(\spe{\rate,\Wm}\) given in 
Definition \ref{def:spherepacking}. Thus
\begin{align}
\notag
\ECXS{\mV}{
\tfrac{1-\sta_{\ind}}{\sta_{\ind}}\left(\ell_{\ind}\RC{\sta_{\ind}}{\Wm}-\hrv_{\ind}\right)}{\fltrn{\ind-1}}
&\!\leq\!\ell_{\ind}\!\spe{\rate-\delta_{1},\Wm\!}\!
\\
\label{eq:LSPBF_16}
&\hspace{1cm}+\!\ell_{\ind}\!\tfrac{2\rate \epsilon}{\rnf^{2}}.
\end{align}
\end{itemize}
Using \eqref{eq:LSPBF_14}, \eqref{eq:LSPBF_15}, and \eqref{eq:LSPBF_16} we get
\begin{align}
\label{eq:LSPBF_17}
\hspace{-.2cm}
0\!\leq\!\ECXS{\mV}{\vrv_{\ind}}{\fltrn{\ind-1}}
&\!\leq\!\ell_{\ind}\!\spe{\rate-\delta_{1},\Wm\!}\!+\!\ell_{\ind}\!\tfrac{2\rate \epsilon}{\rnf^{2}}\!
\end{align}
for all \(\ind\!\in\!\{\!1\!,\!\ldots\!,\!\knd\!\}\)

The analysis for bounding the conditional second moments  
of \(\vrv_{\ind}\)'s is analogous to the one for bounding 
the conditional second moments  of \(\qrv_{\ind}\)'s.
We invoke  \(\sta_{\ind}\geq \rnf-\epsilon\)
instead of \(\sta_{\ind}\leq \rnt+\epsilon\).
\begin{align}
\label{eq:LSPBF_18}
\ECXS{\mV}{(\vrv_{\ind})^{2}}{\fltrn{\ind-1}}
&\leq 16\tfrac{(2+\RC{\sfrac{1}{2}}{\Wm})^{2}\blx^{2}}{(\rnf)^{2}\knd^{2}}.
\end{align}
Applying Lemma \ref{lem:chebyshev},
for \(\amn{\ind}=\ell_{\ind}(\spe{\rate-\delta_{1},\Wm\!}\!+\!\tfrac{2\rate \epsilon}{\rnf^{2}})\)
to the stochastic 
sequence \((\vrv_{1},\smplA_{1}),\ldots,(\vrv_{\knd},\smplA_{\knd})\)
via \eqref{eq:LSPBF_17} we get
\begin{align}
\notag
\PXS{\mV}{\vrv\!\leq\!\blx (\spe{\rate-\delta_{1},\Wm\!}\!+\!\tfrac{2\rate \epsilon}{\rnf^{2}})+\gamma}
&\!\geq\!1\!-\!\tfrac{\sum\nolimits_{\ind=1}^{\knd}\EXS{\mV}{(\vrv_{\ind})^{2}}}{\gamma^{2}},
\end{align}
where \(\vrv\) is defined as
\begin{align}
\label{eq:LSPBF_Vdef}
\vrv
&\DEF\sum\nolimits_{\ind=1}^{\knd} \vrv_{\ind}.
\end{align}
Setting
\(\gamma=8\tfrac{(2+\RC{\sfrac{1}{2}}{\Wm})\blx}{\rnf\sqrt{\knd}}\)
and invoking \eqref{eq:lem:SPBF-delta2} and  \eqref{eq:LSPBF_18} we get
\begin{align}
\label{eq:LSPBF-V-MC-prob}
\PXS{\mV}{\oev_{\mV}}
&\geq \tfrac{3}{4},
\end{align}
where \(\oev_{\mV}\) is defined as
\begin{align}
\label{eq:LSPBF-V-MC-event}
\!\oev_{\mV}
&\!\DEF\!\left\{\dsmpl\!\in\!\smplS\!:\!
\vrv(\dsmpl)\!\leq\!\!\blx (\spe{\rate\!-\!\delta_{1},\!\Wm}\!+\!\delta_{2})
\!+\!\ln\!\tfrac{1}{4\blx^{\knd}}\right\}.
\end{align}
The conditional distribution of \(\prob_{\!\mV}\), and  \(\prob\)
for \(\sta_{\ind}\)'s given \(\smplA_{\ind-1}\) are identical 
for all \(\ind\in\{1,\ldots,\knd\}\) because of \eqref{eq:LSPBF_01}.
Thus \(\vrv(\dsmpl)=\ln\der{\prob_{\!\mV}}{\prob}(\dsmpl)\) and consequently 
\begin{align}
\label{eq:LSPBF-V-MC-bound}
\PX{\set{B}\cap \left\{\vrv\leq \lambda\right\}}
&\geq e^{-\lambda}\PXS{\mV}{\set{B}\cap \left\{\vrv\leq \lambda\right\}}
\end{align}
for any \(\set{B}\in\smplA\) and \(\lambda\in\reals{}\).

As a result of 
\eqref{eq:LSPBF-Q-MC-prob}
and
\eqref{eq:LSPBF-V-MC-prob}
we have
\begin{align}
\label{eq:LSPBF-MC-prob}
\PXS{\mV}{\oev_{\mQ}\cap\oev_{\mV}}
&\geq \tfrac{1}{2},
\end{align}
where \(\oev_{\mQ}\) and \(\oev_{\mV}\) are defined 
in \eqref{eq:LSPBF-Q-MC-event} and \eqref{eq:LSPBF-V-MC-event},
respectively.

\begin{remark}
	If we could show 
	\(\PXS{\mQ}{\mes\!\neq\!\est}\!\approx\!e^{-\blx \rate}\), 
	then we would use 
	\eqref{eq:LSPBF-Q-MC-bound}, \eqref{eq:LSPBF-V-MC-bound}, 
	and \eqref{eq:LSPBF-MC-prob}, 
	to bound the error probability under \(\prob\)--- i.e., 
	to bound \(\Pem{av}\)---
	from below. 
	However, the distribution of \(\out_{1}^{\blx}\) depends on \(\mes\) not only
	under \(\prob\) and \(\prob_{\!\mV}\) but also under \(\prob_{\!\mQ}\) 
	because of \(\sta\)'s. 
	We cope with this issue using a pigeon hole argument.
\end{remark}

\subsection{A Change of Measure Argument together with a Pigeon Hole Argument}\label{sec:pigeonhole}

Let us consider the random variables \(\grv_{1},\ldots,\grv_{\knd}\). 
Let us divide the interval \((0,1]\)
into \(\blx\) intervals of length \(\sfrac{1}{\blx}\). 
Thus for each \(\ind\) in \(\{1,\ldots,\knd\}\),
the value of the random variable \(\grv_{\ind}\) 
will be in only one of the  \(\blx\) intervals 
for each sample point \(\dsmpl\!\in\!\smplS\). 
Thus we get \(\blx^{\knd}\) disjoint \(\knd\)-cubes whose union is \((0,1]^{\knd}\)
for the vector \(\grv_{1}^{\knd}\). 
For each \(\knd\)-cube \(\zeta\), let us  define the event \(\oev_{\zeta}\in\smplA\) as
\begin{align}
\notag
\oev_{\zeta}
&\!\DEF\!\{\dsmpl\!\in\!\smplS: \grv_{1}^{\knd}(\dsmpl)\in\zeta\}.
\end{align}
As a result of \eqref{eq:LSPBF-MC-prob} there exists at least one \(\knd\)-cube \(\zeta^{*}\) satisfying 
\begin{align}
\label{eq:LSPBF-MC-prob-zeta}
\PXS{\mV}{\oev_{\mQ}\cap\oev_{\mV}\cap \oev_{\zeta^{*}}}
&\geq \tfrac{1}{2\blx^{\knd}}.
\end{align}

\begin{figure}[ht]
	\begin{center}
	\begin{tikzpicture}	
	\draw [thick] (0,0) -- (9,0);
	
	\draw[red,thick,densely dotted] (4.4,0) -- (4.4,.6) -- (6.9,.6)-- (6.9,0);
	\draw[red,decorate,decoration={brace,amplitude=6pt},xshift=0.0pt,yshift=0.0pt](4.4,.6) -- (6.9,.6) node [red,midway,yshift=0.3cm] 
	{\footnotesize $\epsilon$};
	\draw[red,decorate,decoration={brace,amplitude=5pt,mirror},xshift=0.0pt,yshift=0.0pt](6.9,0) -- (6.9,.6) node [red,midway,xshift=+0.3cm] 
	{\footnotesize $\tfrac{1}{\epsilon}$};
	
	\draw[blue,dashed] (3,0) --(3,0.5) -- (6,0.5)-- (6,0);
	\draw[blue,decorate,decoration={brace,amplitude=8pt,mirror},xshift=0.0pt,yshift=0.0pt](3,0) -- (6,0) node [blue,midway,yshift=-0.5cm] 
	{\footnotesize $\tilde{\epsilon}$};
	\draw[blue,decorate,decoration={brace,amplitude=4pt},xshift=0.0pt,yshift=0.0pt](3,0) -- (3,.5) node [blue,midway,xshift=-0.3cm] 
	{\footnotesize $\tfrac{1}{\tilde{\epsilon}}$};
	\end{tikzpicture}
	\vspace{-.8cm}
	\caption{A representation of the conditional probability density functions 
		of \(\sta_{\ind}\) given \(\fltrn{\ind-1}\)
		under \(\prob_{\!\mU}\) and \(\prob_{\!\mQ}\),
		which are described in \eqref{eq:LSPBF_01u} and \eqref{eq:LSPBF_01}.
		For all realizations of \(\fltrn{\ind-1}\),
		the conditional probability density function 
		of \(\sta_{\ind}\) under \(\prob_{\!\mU}\) is the same:
		it is equal to \(\sfrac{1}{\tilde{\epsilon}}\) on an interval of length
		\(\tilde{\epsilon}\) and zero elsewhere.
		We represent it with a dashed line in the above figure. 
		For all realizations of \(\fltrn{\ind-1}\), 
		the conditional probability density function 
		of \(\sta_{\ind}\) under \(\prob_{\!\mQ}\)
		is equal to \(\sfrac{1}{{\epsilon}}\) on some interval of length
		\({\epsilon}\) and zero elsewhere, as well.
		However, the starting point of this interval, i.e. \((1-\epsilon)\grv_{\ind}\),
		---and hence the conditional density function \(\sta_{\ind}\) under \(\prob_{\!\mQ}\)
		itself---
		depends on the realization of \(\fltrn{\ind-1}\).
		We represent it with a dotted line in the above figure.}
	\vspace{-.3cm}
	\label{fig:pigeonhole}
	\end{center}
\end{figure} 

Let us assume without loss of generality that 
\begin{align}
\notag
\zeta^{*}
&\!=\!\left(\rnb_{1}\!-\!\tfrac{\rnb_{1}}{\blx},\rnb_{1}\!+\!\tfrac{1-\rnb_{1}}{\blx}\right]\times\cdots\times 
\left(\rnb_{\knd}\!-\!\tfrac{\rnb_{\knd}}{\blx},\rnb_{\knd}\!+\!\tfrac{1-\rnb_{\knd}}{\blx}\right]
\end{align} 
for some \(\rnb_{1},\ldots,\rnb_{\knd}\).
Let us define the probability measure \(\prob_{\!\mU}\) on \((\smplS,\smplA)\)
by setting its marginal on \(\mesS\) to the uniform  distribution
and defining its conditional distributions  as follows:
\begin{align}
\label{eq:LSPBF_01u}
\PCXS{\mU}{\cset}{\dmes,\dsta_{1}^{\ind-1},\dout_{1}^{\tin_{\ind-1}}}	
&\!=\!\tfrac{1}{\tilde{\epsilon}}
\int_{(1-\tilde{\epsilon})\rnb_{\ind}}^{\rnb_{\ind}+\tilde{\epsilon}(1-\rnb_{\ind})} \TCIND{\cset}{\dsta} \dif{\dsta}
\end{align}
for all \(\cset\!\in\!\rborel{\staS_{\ind}}\),
where \(\tilde{\epsilon}=\epsilon+\tfrac{1-\epsilon}{\blx}\)
and
\begin{align}
\label{eq:LSPBF_02u}
\PCXS{\mU}{\dout_{\tau_{\ind}}^{\tin_{\ind}}}{\dmes,\dsta_{1}^{\ind},\dout_{1}^{\tin_{\ind-1}}}	
&\!=\!
\prod\nolimits_{\tin=\tau_{\ind}}^{\tin_{\ind}}\!
\qmn{\dsta_{\ind},\Wm}(\dout_{\tin})
\end{align}
for all \(\dout_{\tau_{\ind}}^{\tin_{\ind}}\!\in\!\outS_{\tau_{\ind}}^{\tin_{\ind}}\).

Comparing 
\eqref{eq:LSPBF_01u}
and
\eqref{eq:LSPBF_02u}
describing the conditional distributions of \(\prob_{\!\mU}\)
with 
\eqref{eq:LSPBF_01}
and
\eqref{eq:LSPBF_02q}
describing the conditional distributions of \(\prob_{\!\mQ}\),
we can conclude that 
\begin{align}
\label{eq:LSPBF-PH-1}
\PXS{\mQ}{\set{B}\cap \oev_{\zeta^{*}}}
&\leq 
(\tfrac{\tilde{\epsilon}}{\epsilon})^{\knd}
\PXS{\mU}{\set{B}}
\end{align}
for any \(\set{B}\in\smplA\). 

Since the distribution of \(\out_{1}^{\blx}\) does not depend on \(\mes\)
under \(\prob_{\!\mU}\), we have
\begin{align}
\notag
\PXS{\mU}{\mes\!=\!\est}
&\leq \tfrac{1}{\lceil e^{\blx\rate} \rceil}.
\end{align}
Invoking \eqref{eq:LSPBF-PH-1} for 
\({\set B}=\oev_{\mQ}\cap\oev_{\mV}\cap \{\mes\!=\!\est\}\)
we get
\begin{align}
\notag
\PXS{\mQ}{\oev_{\mQ}\cap\oev_{\mV}\cap\oev_{\zeta*}\cap \{\mes\!=\!\est\}}
&\leq (\tfrac{\tilde{\epsilon}}{\epsilon})^{\knd} e^{-\blx\rate}.
\end{align}
If we use \eqref{eq:LSPBF-Q-MC-event}
and  \eqref{eq:LSPBF-Q-MC-bound}
for \(\lambda\!=\!\blx \rate\!-\!\ln 4\!-\!\knd\!\ln(\blx+\tfrac{1}{\epsilon})\) 
and recall \(\tilde{\epsilon}=\epsilon+\tfrac{1-\epsilon}{\blx}\) we get
\begin{align}
\notag
\PXS{\mV}{\oev_{\mQ}\cap\oev_{\mV}\cap\oev_{\zeta*}\cap \{\mes\!=\!\est\}}
&\leq 
\tfrac{e^{\!\blx \rate}}{4}
(\blx+\tfrac{1}{\epsilon})^{-\knd}
(\tfrac{\tilde{\epsilon}}{\epsilon})^{\knd} e^{-\blx \rate}
\\
\notag
&= \tfrac{1}{4} (\tfrac{1}{\epsilon\blx+1})^{\knd}(\tfrac{\epsilon\blx+(1-\epsilon)}{\blx})^{\knd}
\\
\notag
&\leq \tfrac{1}{4\blx^{\knd}}. 
\end{align}
Then as a result of \eqref{eq:LSPBF-MC-prob-zeta},
\begin{align}
\notag
\PXS{\mV}{\oev_{\mQ}\cap\oev_{\mV}\cap\oev_{\zeta*}\cap \{\mes\!\neq\!\est\}}
&\geq \tfrac{1}{4\blx^{\knd}}.
\end{align}
If we use \eqref{eq:LSPBF-V-MC-event}
and  \eqref{eq:LSPBF-V-MC-bound}
for 
\(\lambda\!=\!\blx (\spe{\rate\!-\!\delta_{1},\!\Wm}\!+\!\delta_{2})
+\ln\!\tfrac{1}{4\blx^{\knd}}\), then
we get
\begin{align}
\notag
\PX{\oev_{\mQ}\cap\oev_{\mV}\cap\oev_{\zeta*}\cap \{\mes\!\neq\!\est\}}
&\geq e^{-\blx \left[\spe{\rate-\delta_{1},\Wm}+\delta_{2}\right]}.
\end{align}
Then \eqref{eq:lem:SPBF} holds because \(\Pem{av}\!=\!\PX{\mes\!\neq\!\est}\).

\subsection{Proof of Theorem \ref{proposition:SPBF}}\label{sec:proof:proposition:SPBF}
If  \(\epsilon_{\blx}=\tfrac{\rnf\wedge(1-\rnt)}{\blx}\)
and \(\knd_{\blx}=\lfloor \blx^{\sfrac{2}{3}}\rfloor\),
then there exists an \(\blx_{0}\)
for which 
\(\delta_{1}\) defined in \eqref{eq:lem:SPBF-delta1}
and
\(\delta_{2}\) defined in \eqref{eq:lem:SPBF-delta2}
satisfy \(\delta_{1}\vee \delta_{2} \leq\tfrac{2\ln \blx}{\blx^{1/3}}\)
for all \(\blx\geq\blx_{0}\).
Then for any \(\blx\geq\blx_{0}\),
the hypotheses of Lemma \ref{lem:SPBF} is satisfied 
by
any code satisfying the hypotheses of Theorem \ref{proposition:SPBF} 
and Theorem \ref{proposition:SPBF}
follows from Lemma \ref{lem:SPBF}.

 \section{Discussion}\label{sec:conclusion}
We have proved  both
the non-asymptotic  SPB given in Lemma \ref{lem:SPBF}
and
the asymptotic SPB given in Theorem \ref{proposition:SPBF} 
for codes on DSPC with feedback
in order keep the analysis as simple as possible.
Nevertheless, the proofs work, essentially, as is for codes on 
finite output set stationary product channels with feedback, as
well. 
Augustin, on the other hand, stated his asymptotic result
\cite[Thm. 41.7]{augustin78} for codes on  
finite input set stationary product channels with feedback.

In a general stationary product channel with feedback,
the stochastic matrix \(\!\Wm\!\in\!\pdis{\outS|\inpS}\) is replaced 
by a transition probability \(\!\Wm\!\in\!\pmea{\outA|\inpA}\),
see \cite[{Definition 8}]{nakiboglu19B}.
In order to generalize Lemma \ref{lem:SPBF}
to stationary product channels with feedback,
we first need to prove Lemma \ref{lem:intermediate}.
That can be done rather easily by assuming
\begin{align}
\label{eq:equivalent}
\lim\nolimits_{\rno \uparrow1}\tfrac{1-\rno}{\rno}\RC{\rno}{\Wm}
&\!=\!0.
\end{align}
The challenge lies in the construction of
probability measures \(\prob\), \(\prob_{\!\mV}\),
and \(\prob_{\!\mQ}\) 
and in determining the functions \(\gX_{1},\ldots,\gX_{\knd}\):  
we need to show that expressions given in 
\eqref{eq:LSPBF_01}, \eqref{eq:LSPBF_02w}, \eqref{eq:LSPBF_02q},
\eqref{eq:LSPBF_02v}  define Borel functions for an
appropriate choice of the functions \(\gX_{1},\ldots,\gX_{\knd}\)
and that the same choice ensures 
\eqref{eq:LSPBF_06}, \eqref{eq:LSPBF_07}, and \eqref{eq:LSPBF_08}.
The other parts of the proof of Lemma \ref{lem:SPBF} and 
the proof of Theorem \ref{proposition:SPBF} work  as is.
Augustin has asserted in 
\cite[Corollary 41.9]{augustin78}
that his proof sketch works for
codes on stationary product channels with 
feedback whose component channel \(\Wm\) satisfies\footnote{\cite[Corollary 41.9]{augustin78} assumes 
	the conditional weak compactness,
	which is just another way of assuming \eqref{eq:equivalent}.
	The condition given in  \eqref{eq:equivalent} 
	and being conditionally weakly compact
	---i.e., having compact closure in the weak topology---
	are equivalent by 
\cite[{Lemma 24-(d)}]{nakiboglu19A}.
} \eqref{eq:equivalent}.

The SPBs are customarily stated for the list decoding, 
e.g. \cite[(1.4) and Thm. 2]{shannonGB67A};
however, we have confined our discussion to the case without 
the list decoding for the sake of simplicity. 
Nevertheless, both Lemma \ref{lem:SPBF} and Theorem \ref{proposition:SPBF}
can be extended to the list decoding case in a straightforward way.

Recently, we have proposed another proof for the SPB for codes on DSPCs with feedback
\cite[Thm. 3]{nakiboglu19B}
and generalized it to codes on (possibly non-stationary) DPCs
with feedback
\cite[Thm. 4]{nakiboglu19B}.
It seems analogous generalizations are possible for 
Theorem \ref{proposition:SPBF} and Lemma \ref{lem:SPBF} under similar hypotheses.
A natural next step would be considering 
codes on the cost constrained
stationary memoryless channels with feedback. 
Under certain hypothesis, it is possible to 
establish the SPB using the proof technique applied here, 
but we are not aware of a general proof that will work
for all cost constrained stationary discrete memoryless 
channels with feedback.

\appendix
\begin{proof}[Proof of Lemma \ref{lem:chebyshev}]
\hypertarget{appendix}{For} all \(\tin\) in \(\{1,\ldots,\blx\}\), let \(\sta_{\tin}\) be 
\begin{align}
\notag
\sta_{\tin}
&\DEF\inp_{\tin}-\ECX{\inp_{\tin}}{\smplA_{\tin-1}}.
\end{align}	
Then \(\sta_{\tin}\) is	an \(\smplA_{\tin}\)-measurable random variable.
Furthermore,
\(\ECX{\sta_{\tin}\sta_{\tau}}{\fltrn{\tau}}\!=\!\ECX{\sta_{\tin}}{\fltrn{\tau}}\sta_{\tau}\!=\!0\)
for any \(\tau\) and \(\tin\) satisfying
	\(1\leq \tau<\tin\leq\blx\). Thus
	\begin{align}
	\notag
	\EX{\left(\sum\nolimits_{\tin=1}^{\blx}\sta_{\tin}\right)^{2}}
	&\!=\!\sum\nolimits_{\tin=1}^{\blx}\!\EX{(\sta_{\tin})^{2}}
	\\
	\notag
	&\!=\!\sum\nolimits_{\tin=1}^{\blx}\!\EX{\ECX{\left(\inp_{\tin}-\ECX{\inp_{\tin}}{\smplA_{\tin-1}}\right)^{2}}{\smplA_{\tin-1}}}
	\\
	\notag
	&\!\leq\!\sum\nolimits_{\tin=1}^{\blx}\!\EX{\ECX{\left(\inp_{\tin}\right)^{2}}{\smplA_{\tin-1}}}
	\\
	\notag
	&\!=\!\sigma^{2}.
	\end{align}
	Then the Chebyshev's inequality implies 
	\begin{align}
	\label{eq:chebyshev-1}
	\PX{\sum\nolimits_{\tin=1}^{\blx}\sta_{\tin}<\gamma}&\geq 1- \tfrac{\sigma^{2}}{\gamma^2}.
	\end{align}
	On the other hand,
	\(\sum\nolimits_{\tin=1}^{\blx}\sta_{\tin}\geq \sum\nolimits_{\tin=1}^{\blx}(\inp_{\tin}-\amn{\tin})\)
	holds with probability one
	because \(\ECX{\inp_{\tin}}{\fltrn{\tin-1}}\!\leq\!\amn{\tin}\) 
	with probability one by the hypothesis. Thus
	\begin{align}
	\label{eq:chebyshev-2}
	\PX{\sum\nolimits_{\tin=1}^{\blx}\sta_{\tin}<\gamma}
	&\leq \PX{\sum\nolimits_{\tin=1}^{\blx}(\inp_{\tin}-\amn{\tin})<\gamma}.
	\end{align}
	\eqref{eq:lem:chebyshev} follows from \eqref{eq:chebyshev-1} and \eqref{eq:chebyshev-2}.
\end{proof}

\section*{Acknowledgment}
The author would like to thank Fatma Nakibo\u{g}lu and Mehmet Nakibo\u{g}lu for 
their hospitality; this work simply would not have been possible without it.  
The author completed his initial study confirming the soundness of Augustin's 
method at UC Berkeley in 2012 fall while he was working as a postdoctoral 
researcher. The author would like to thank 
Imre Csisz\'{a}r for pointing out Agustin's work at Austin in 2010 ISIT,
Harikrishna R. Palaiyanur for sending him Augustin's manuscript \cite{augustin78},
and the reviewers for their comments and suggestions on the manuscript.
\bibliographystyle{ieeetr} 
\bibliography{main}
\end{document}